%% file: networkarchive2.tex
\documentclass[12pt,transaction, draftclsnofoot,onecolumn]{IEEEtran}

\usepackage{amsthm}
\usepackage{amsfonts}
\usepackage{amssymb}
\usepackage{hyperref}
\usepackage{graphicx}
\usepackage{enumerate}
\usepackage{amsmath}
\usepackage{color}
\usepackage{algpseudocode}
\usepackage{algorithm}
\usepackage{array}
\usepackage{soul}

\input{macros.tex}

\newtheorem{theorem}{Theorem}
\newtheorem{lemma}{Lemma}
\newtheorem{proposition}{Proposition}
\newtheorem{corollary}{Corollary}

\begin{document}

\title{Reputational Learning and Network Dynamics}

\author{\space Simpson~Zhang
        and~Mihaela~van~der~Schaar
        \thanks{Document Date: June 2016}\thanks{The authors are indebted to Jie Xu, Peng Yuan Lai, William Zame, and Moritz Meyer-ter-Vehn for valuable assistance. This paper also benefited from discussions with Matt Jackson and seminar participants at the 2015 NEGT Conference and the 2015 SWET conference. The authors gratefully acknowledge financial support from the ONR.}\thanks{
        Simpson Zhang: simpsonzhang@ucla.edu,
        Mihaela van der Schaar: mihaela@ee.ucla.edu}
}

\maketitle
\begin{abstract}
\boldmath
In many real world networks agents are initially unsure of each other's qualities and must learn about each other over time via repeated interactions. This paper is the first to provide a methodology for studying the dynamics of such networks, taking into account that agents differ from each other, that they begin with incomplete information, and that they must learn through past experiences which connections/links to form and which to break. The network dynamics in our model vary drastically from the dynamics in models of complete information. With incomplete information and learning, agents who provide high benefits will develop high reputations and remain in the network, while agents who provide low benefits will drop in reputation and become ostracized. We show, among many other things, that the information to which agents have access and the speed at which they learn and act can have a tremendous impact on the resulting network dynamics. Using our model, we can also compute the \textit{ex ante} social welfare given an arbitrary initial network, which allows us to characterize the socially optimal network structures for different sets of agents. Importantly, we show through examples that the optimal network structure depends sharply on both the initial beliefs of the agents, as well as the rate of learning by the agents. Due to the potential negative consequences of ostracism, it may be necessary to place agents with lower initial reputations at less central positions within the network.
\end{abstract}


\section{Introduction}

Networks are pervasive in all areas of society, ranging from financial networks to organizational networks to social networks. And an important feature of many real world networks is that agents do not fully know the characteristics of others initially and must learn about them over time. For instance a bank learns about the credit-worthiness of a new borrower, a worker in a firm learns about the ability of a coworker, and a buyer learns about the product quality of a supplier. Such learning can strongly affect the resulting shape of the network. As agents receive new information, they can revise their beliefs about other agents, update their linking decisions, and cause the network to evolve as a result. To properly analyze such network evolution, it is crucial to understand the exact mechanism by which learning impacts network dynamics. 

The impact of agent learning on network evolution has not been well studied in the existing literature. A large network science literature analyzes the effect of learning on fixed networks that have already formed (see Scott (2012)).  A smaller microeconomics literature\footnote{See the overview in Jackson (2010) for instance.} studies the formation of networks - but makes very strong assumptions (e.g., homogeneous agents/entities, complete information about other agents).  Neither the network science literature nor the microeconomics literature has so far taken into account that agents behave strategically in deciding what links to form/maintain/break \textit{and} that they also begin with incomplete information about others, so they must learn about others through their interactions. As a result, neither network science nor microeconomics provides a complete framework for understanding, predicting and guiding the formation (and evolution) of real networks and the consequences of network formation.

The overarching goal of this research paper is to develop such a framework. An essential part of the research agenda is driven by the understanding that individuals within a network are heterogeneous - some workers are more productive than others, some friends are more helpful than others, and some borrowers are more creditworthy than others. Furthermore, these characteristics are not known in advance but must be learned over time via repeated interactions. The rate of learning itself may also be strongly influenced by the network structure: agents engaged in more interactions are likely reveal more information about themselves. 

As a motivating example, consider a group of financial institutions that are linked together in a financial network\footnote{Our model can also be applied to a wide range of other networks, such as organizational networks, social networks, or expertise networks. We discuss some implications for these settings as well throughout the paper.}. These financial institutions provide benefits to each other by engaging in mutually beneficial trading opportunities, such as providing each other with liquidity or investing in joint ventures\footnote{As in the model of Erol (2015).}. High quality institutions are likely to develop a high realized quality of assets from these joint interactions, while low quality institutions are likely to develop a low realized quality of assets. Each institution only continues to link with another institution (over time) if the counterparty is believed to be of sufficiently high quality. As time progresses, the institutions \textit{observe} the actions of their counterparties, \textit{update} their beliefs about the quality of each counterparty, and \textit{change} their linking decisions as a result. In this way, learning by the financial institutions causes the network topology to evolve over time. The network topology also impacts the rate of learning, as an institution can learn through both its own interactions with counterparties, as well as by monitoring interactions of a counterparty with its other counterparties.  Since institutions with more connections interact with more counterparties, such institutions will reveal more information about themselves to their neighbors over time. As a result, while having more connections opens an institution up to more beneficial opportunities, it also carries the risk of causing the institution to be shut out of the financial network more quickly if it starts losing asset value, as in the case of Lehman Brothers due to its exposures to the subprime mortgage market during the 2008 financial crisis.

Our model takes into account the features of the previous example: agents behave strategically, begin with incomplete information about each other, and must learn through continued interactions which connections to form and maintain and which to break. We consider a continuous time model with a group of agents who are linked according to a network and who send noisy flow benefits to their neighbors. The benefits that agents provide could be interpreted for instance as the benefits that financial institutions derive from providing liquidity to each other or from diversifying risk with each other's specialized assets.  Each agent is distinguished by a fixed quality level which determines the average value of the flow benefits it produces. Agents observe all the benefits that their neighbors produce, and they update their beliefs about a neighbor's quality via Bayes rule. Neighbors with more connections will reveal more information about themselves over time. Agents will maintain links with neighbors that provide high benefits, but will cut off links with neighbors that provide low benefits. The network evolves as agents learn about each other and update their beliefs. Since the number of links an agent has affects the rate of learning about that agent, the rate of learning about an agent changes as the network changes, leading to a co-evolution of network topology and information production.

Our model is highly tractable and allows us to completely characterize network dynamics and give explicit probabilities that the network evolves into various configurations. In addition, we are able to characterize the entire set of possible stable networks and analytically compute the probability that any single stable network emerges. This allows for predictions regarding which types of stable networks are likely to emerge given an initial network.

We also study the implications that learning has on the social welfare and efficiency of a network. Our results show that learning has a beneficial aspect: agents that are of low true quality are likely to produce low signals and will eventually be ostracized from the network. Learning also has a harmful aspect: even high true quality agents may produce an unlucky string of bad signals and so be forced out of the network. Moreover, even having low true quality agents leave the network can reduce overall social welfare. A marginally low quality agent may harm its neighbors slightly, but it also receives a large benefit if its neighbors are of very high quality. Therefore if the low quality agent leaves the network, the overall social welfare would actually decrease. The issue here is that agents only care about the benefit their neighbors are providing them, but not the benefit they are providing their neighbors. This results in a negative externality every time a link is severed\footnote{The negative effects of ostracism can be particularly acute in financial networks during times of distress in which banks get shut out of funding, as is the case of a liquidity freeze. Ostracism has also been demonstrated in a wide variety of social settings  in the social psychology literature. We discuss this literature and our model's implications in the Literature Review section.}. In many situations, the negative effects of learning outweigh the positive effects, so on balance learning is actually harmful. In particular, increasing the learning rate about marginal agents whose neighbors are high quality agents is bad, because forcing the marginal quality agent out of the network sacrifices the social benefit of the link to the high quality agent. However, increasing the rate of learning about a marginal quality agent whose neighbors are also marginal quality agents is good, because more information will be revealed about that marginal quality agent, allowing its neighbors to more quickly sever their links to it. The impact of learning can therefore be either positive or negative depending on the specific network.

Our welfare results have important implications for network planning and are useful in a diverse range of settings, such as in guiding the formation of networks by the policies of a financial regulator, human resources department, online community, etc. Due to the varying effects of learning, we show that the optimal network design will be quite different for different groups of agents. For instance, when agents all have high initial reputations, the optimal network design allows all agents to be connected (so that agents can benefit fully from their repeated interactions). On the other hand, if some agents have low initial reputations, then allowing all agents to connect is not optimal, and it will be desirable to constrain the network by isolating low reputation agents from each other.  If such agents did link, they would both send more information about themselves through this link, causing themselves to be ostracized more quickly. Each agent, as well as the overall network, could then be worse off through the formation of this link due to the faster learning caused by the link. In some cases, a star or a core-periphery network connectivity structure would generate higher social welfare than a complete network even when all agents have initial expected qualities higher than the linking cost. Such a situation arises for instance if there are two separate groups of agents, one group with very high reputation and the other group with moderate reputation. By placing the high reputation agents in the core and the moderate reputation agents in the periphery, the high reputation agents are able to produce large benefits for the network, and the potential harm from the moderate reputation agents is minimized\footnote{This provides a new reputational reason for the benefits of a core-periphery network, in contrast to other, non-informational, reasons that have been proposed in the networks literature.}. 

Finally, we consider four extensions of our model that allow for even richer network dynamics and learning. In the first extension, we allow the mechanism designer to provide the agents with a subsidy that encourages linking\footnote{For instance a financial regulator could guarantee transactions within a financial network to make them less risky.}. The effect of such a subsidy is to promote the amount of experimentation done by the agents, and we show that a sufficiently large subsidy can always improve overall social welfare because of this. In the second extension, we allow for agents with high enough reputations to form new links with each other, and we show that social welfare will be increased when the linking threshold is high enough. In the third extension, we allow new agents to enter the network over time, and we consider the optimal time at which new agents should arrive. We show that all agents should be allowed to enter the network eventually, but delayed entry is desirable in certain networks to protect the reputations of vulnerable incumbent agents. Lastly, in the fourth extension we allow for agents that have been ostracized in the past to re-enter the network after a set period of time, and we show that the negative effects of learning can be mitigated if re-entry occurs frequently enough.

\nocite{acemoglu2011bayesian} \nocite{acemoglu2011opinion} \nocite{bala2000noncooperative}
\nocite{babus2015endogenous}
\nocite{blasques2015dynamic}
\nocite{chang2015endogenous}
\nocite{gale2003bayesian} \nocite{galeotti2010law} \nocite{galeotti2010network} \nocite{galeotti2006network} \nocite{goeree2009search} \nocite{golub2010naive} \nocite{golub2012homophily}
\nocite{hollifield2014bid} \nocite{erol2015network} \nocite{erol2014network}
\nocite{farboodi2015intermediation} \nocite{jackson2010social} \nocite{jackson1996strategic} \nocite{jackson2014games} \nocite{jacobson1993myopic} \nocite{kozlowski2013work}
\nocite{li2014dealer}
\nocite{mizik2010theory}
\nocite{neklyudov2015endogenous} \nocite{robinson2012invisible} \nocite{scott2012social} \nocite{song2015dynamic} \nocite{vandynamic} \nocite{van2015acquaintances} \nocite{vega2006building}
\nocite{wang2016core} \nocite{wang1997boundary} \nocite{watts2001dynamic} \nocite{watts1998collective} \nocite{wesselmann2013we} \nocite{williams1997social}

\section{Literature Review}

\subsection{Relation to Theoretical Networks Literature}
Our paper represents a novel contribution to the network formation literature, by being among the first to consider incomplete information and learning in networks, as well as by providing a tractable model that allows for the computation of many properties, including the \textit{ex ante} social welfare, of different network topologies. Other papers in the network literature have usually studied network dynamics only in settings of complete information when agents perfectly know each other's qualities. For example, the papers by Jackson and Wolinsky (1996), Bala and Goyal (2000), Watts (2001), and Galeotti and Goyal (2010) all consider networks where the agents have complete information. In these models, agents are aware of the exact qualities of all other agents and there is no learning. The network dynamics arise instead from externalities and indirect benefits between agents that are not directly linked. When one link is formed or severed, the benefits produced by other links changes as well, which causes the other agents to sever or form their own links in a chain reaction. For some networks, such as communication networks, these indirect benefits seem important, as an agent who has many high quality neighbors will likely be able to transmit higher quality information as well. However, in other networks such as friendship networks these indirect benefits are less relevant and it is the specific quality of each individual agent that is the most relevant. This is especially applicable in situations where a new group of agents are meeting for the first time, and they learn about each other through mutual interactions. We argue that the network dynamics in such situations are more greatly dependent on reputational effects and than on changes in the values of indirect benefits.

We do not assume any indirect benefits in our model and focus instead on the dynamics resulting from incomplete information and learning. Agent learning strongly influences the network formation process in a way that would not arise with complete information. Agents that send good signals will develop high reputations and remain in the network, whereas agents that send bad signals will develop low reputations and eventually become ostracized by having their neighbors cut off links. The rate of learning about an agent's quality affects how quickly the network evolves and has a strong effect on the resulting social welfare. With complete information however, such dynamics would not occur because agents would know each other's qualities perfectly at the onset. For instance, Watts (2001) considers a dynamic network formation model where agents form links under complete information. When there are no indirect benefits between agents in that paper's model, each agent would make a one time linking decision with every other agent and never update its choice later on. But with learning, agents may change their linking choices by breaking off links with neighbors that consistently produce low benefits. Incomplete information causes links to fluctuate dynamically over time as new information arrives and beliefs are updated, instead of staying static as in the complete information case. We propose that such effects are key and even the main driver of dynamics when a group of agents are meeting for the first time and forming a network with each other.

In addition, the tractability of our model allows us to explicitly compute the social welfare for different network structures even under incomplete information. This tractability arises from the use of continuous time diffusion processes in our model, which allows for closed form equations of the probabilities that different networks emerge. In contrast other networks papers such as Jackson and Wolinsky (1996) and Bala and Goyal (2000) use discrete time models that do not allow for such clean closed form expressions.  While these other papers analyze the efficiency properties of a given fixed network, our welfare results are much stronger and allow the network to evolve endogenously over time as agents learn and update their linking decisions. This enables us to compare the \textit{ex ante} optimality of different initial network structures, as well as provide general results for when specific network structures are optimal. For instance, we show that when the rate of learning in the network is either very slow or very fast, a complete network is optimal if the agent's initial expected qualities are all higher than the cost of maintaining a link. But when learning is at an intermediate rate, it may be optimal to prevent vulnerable agents from connecting, even if their initial expected qualities are higher than the linking cost, due to the negative externalities associated with reputational effects. Such a result cannot arise under complete information, where if agent's qualities are all perfectly known it would be strictly better for all of them to be linked initially.

This paper is also tied to the literature on observational learning in networks. Works such as Golub and Jackson (2010), Acemoglu et al (2011), and Golub and Jackson (2012) analyze observational learning in social networks. In these models there is a fixed exogenous network on which the agents interact, and the agents learn about an exogenous state of the world through this network by observing the actions of neighbors. These papers provide results regarding the speed and accuracy of the observational learning that can be achieved by agents connected through different types of networks. Our paper is significantly different from this literature because agents learn about other agents' \textit{qualities} instead of an \textit{exogenous} state of the world. As such, agents will wish to update their linking decisions over time as their beliefs about the agents with whom they are connected with change. The network and learning \textit{co-evolve}, causing the network structure to evolve \textit{endogenously}. 

Vega-Redondo (2006) focuses on the issue of moral hazard and monitoring, and it considers the diffusion of information about agent actions across a network. It assumes that players engage in bilateral prisoner's dilemma games. Information about player actions diffuses through the network, and agents are able to sustain cooperation through punishing defectors. More densely connected networks allow for faster information transmission and can therefore sustain higher levels of cooperation. The paper analyzes how the structures of the networks that emerge is affected by the transmission of information, and it shows through simulations and mean-field analysis that the inclusion of network based information can increase network density. Our work instead focuses on the issue of adverse selection and on learning about agent types. We show that more information can be harmful for welfare because it leads to greater ostracization among agents. The tractability of our model also allows us to consider the social welfare generated across the entire path of network evolution, as opposed to the welfare of the long run average network. We are therefore able to address issues of network design, and we characterize the optimal network structure under different environments. We also provide simulations which highlight our main results and show the social welfare of different network structures.

A related networks paper that involves learning with adverse selection is Song and van der Schaar (2015). Like us, this paper also considers learning by agents about the types of other agents within a network, and it shows how incomplete information and the learning process can lead to a wide variety of network structures and dynamics. However, this paper considers a discrete time model and incorporates a simplified learning process in which information is revealed \textit{immediately}, after a single interaction. On the contrary, in our model information is revealed gradually, and the linking decisions and learning occur simultaneously. Since learning takes place gradually instead of instantaneously, we are able to analyze how the precise rate of learning affects network dynamics and social welfare. And very importantly, our model allows the network structure itself to impact the rate of learning about agents. This assumption is realistic as learning is often affected by the number of connections an agent has within the network. We show that it has strong implications and necessitates the need for careful planning by a network designer to properly control the learning done by agents. In addition, our use of continuous time allows our model to be more tractable and able to provide explicit characterizations of the social welfare of different network structures.

\subsection{Relation to Financial Networks Literature}
Our paper is also related to the growing financial networks literature. There have been numerous recent papers which seek to explain the prevalent core-periphery structure of financial networks. Such core-periphery financial network structures have been well documented empirically in a variety of markets, such as for municipal bonds (Li and Sch{\"u}rhoff 2014) and securitization (Hollifield et al 2014). The theoretical papers of Chang and Zhang (2015), Farboodi (2015), Neklyudov and Sambalaibat  (2015), Babus and Hu (2015), and Wang (2016) all propose models that seek to explain the prevalence of core-periphery networks. These papers show that features such as various forms of dealer heterogeneity can result in core-periphery type structures. 

However most of these papers operate in complete information settings where the types of other agents are directly observable. And the papers that do consider incomplete information focus instead on learning through investment in information gathering (regarding debt repayment) rather than on learning through interactions which are affected by the network structure itself. For instance, Babus and Hu (2015) shows that since star networks can allow for efficient mutual monitoring by financial institutions, they also lead to more efficient trading. Blasques et al (2015) show that the benefits that a core-periphery network provides leads to greater stability over time. 

Our paper also provides a justification for the multitude of real world core-periphery networks, as we show that such core-periphery networks maximize social welfare in certain networks where agents vary in reputation. However, our result occurs are driven by the presence of reputational forces, unlike the previous papers. In our model, a core-periphery network lowers the reputational risks for vulnerable low-reputation agents, and can thus prevent them from being shut out of the financial network as quickly.

Furthermore the setting of our paper is different from the setting of the other papers. The papers that consider complete information are more relevant for longer time frames and stable financial market conditions where informational uncertainty about counterparties is low. We view our model instead as describing a short time period with great uncertainty. For instance in the aftermath of a financial crisis, banks are very unsure of the solvency of other banks due to the difficulty of assessing the quality of their assets. In such situations, banks will be hesitant to trade with each other and will carefully attempt to learn  the solvency of other institutions through observations of repayments. Thus each bank's reputation evolves over time. Banks that obtain low reputations may get shut out of the funding market entirely during liquidity runs, as was the case during the collapse of Lehman Brothers in the recent financial crisis. It is important for a financial regulator to carefully structure the trading network and control the interactions so that such situations can be mitigated.

\subsection{Relation to Social Ostracism Literature}
Finally, we note that our model also has important implications for social and organizational networks. Our results about the negative externalities of reputational learning highlight the damaging impacts of ostracism found in the social psychology literature. Social ostracism is a prevalent force that has been well documented in the social psychology literature in numerous settings ranging from online interactions to office workplaces. As Williams and Sommer (1997) state, ``Social ostracism is a pervasive and ubiquitous phenomenon." In this literature, ostracism can also occur when an agent's perceived quality drops too low, and will have harmful effects on the agent itself. As the paper by Wesselman et al (2013) notes, ``Ostracism is a common, yet painful social experience...Individuals who do not fit the group's definition of a contributing member may find themselves a likely candidate for \textit{punitive} ostracism". That paper shows the occurrence of ostracism via an online experiment, where agents differ in their ability to play a game. Agents who play badly became ostracized by the others. This effect is similar to our model, where agents who are learned to be of low quality are ostracized from the network.

Ostracism can also occur in workplaces, as some employees may be ostracized by their coworkers. Robinson et al (2012) notes that ``not only are such experiences extremely painful, but under some circumstances they can have an even greater negative impact than other harmful workplace behaviors such as aggression and harassment." It is therefore important for companies to consider the harmful effects of ostracism that can occur through workplace interactions. We provide guidelines for minimizing the negative effects of ostracism through placing lower reputation agents in less central positions of the network.
\section{Model}
\subsection{Overview}
We consider an infinite horizon continuous time model with a finite set of agents denoted by $V = \{1,2,...,N\}$. At every moment in time, the agents choose which other agents to link with, and a link is established only under mutual consent. These choices are made subject to an underlying network constraint $\Omega = \{\omega_{ij}\}$ that specifies the pairs of agents that are able to link with each other\footnote{This network constraint $\Omega$ may arise from the specific interests/desires of the agents regarding who they want to link with, or from potential physical/geographical constraints that limit agents from linking. It may also be planned, e.g. through the policies of a financial regulator for a network of financial institutions, or by the human resources department in a company for a network of employees.}. For each pair of agents $\omega_{ij} = 1$ if agents $i$ and $j$ can connect with each other and $\omega_{ij}=0$ otherwise. We call agents $i$ and $j$ neighbors if they can connect. Initially (time $t = 0$), agents are linked according to a network $G^0 = \{g^0_{ij}\} \subseteq \Omega$. As the network will change over time, we denote $G^t$ as the network at time $t$. Moreover, we let $k^t_i = \sum\limits_{j}g^t_{ij}$ be the number of links that agent $i$ has at time $t$, and we let $K^t_i$ denote the set of neighbors of agent $i$ at time $t$.

Agents receive flow payoffs from each link equal to the benefit of that link minus the cost. Each agent $i$ must pay a flow cost $c$ for each of its links that is active. Hence, at time $t$, agent $i$ pays a total cost of $k^t_i c$ for all its links. Agents also obtain benefits from their links, depending on their linked neighbors' qualities $q_i$. However each agent's true quality is initially unknown to all agents, and we do not require that agents know their own qualities. At the start of the model, each agent $i$'s quality $q_i$ is drawn from a commonly known normal distribution $\mathcal{N}(\mu_i, \sigma^2_i)$ with $\mu_i > c$. Both the mean and the variance are allowed to vary across agents, and several of our results below will utilize this heterogeneity. Agent $i$ generates a different noisy benefit $b_{ij}(t)$ for each agent $j$ that is linked to it, and these benefits follow a diffusion $db_{ij}(t) = q_i dt + \tau_i^{-1/2} d Z_{ij}(t)$, where the drift is the true quality $q_i$ and the variance depends on $\tau_i$, an exogenous parameter we call the \textit{signal precision} of agent $i$\footnote{We can think of the signal precision as representing how much information the agent reveals about itself in each interaction, with a higher precision corresponding to more information. It could depend on the type of interaction with the agent (e.g. close partnerships or chance encounters), or factors like the agent's personality.}. $Z_{ij}(t)$ is a standard zero-drift and unit variance Brownian motion, and represents the random fluctuations in the benefits of each interaction. $Z_{ij}(t)$ is assumed to be independent over all $i$ and $j$, and therefore all the benefits produced by agent $i$ are conditionally independent given $q_i$. We assume that all the benefits that agent $i$ produces are observed by all the neighbors of $i$, which ensures that agent $i$'s neighbors all have the same beliefs about $i$ at any point in time (information is locally public among agent $i$'s neighbors)\footnote{This is an important assumption to maintain the tractability of the model. It can be interpreted, for instance, through an online expertise network where the output of agent $i$ is public, so that all neighbors of agent $i$ can judge the benefit that $i$ has provided to all its links. Or in an offline setting, we could assume that the neighbors of agent $i$ are continuously discussing the benefits they have received from $i$ with all other neighbors of $i$, so that the neighbors maintain the same beliefs. For most of our results, the information does not need to be fully public; the information regarding agent $i$ needs only be available to all the direct neighbors of agent $i$.}. For each agent $i$, we define the agent's benefit history as the history of all previous benefits, $\mathcal{H}^t_i = \{b^{t'}_{ij}\}^t_{t'=0}$. 

We assume that agents are myopic, and they thus consider only the current flow benefit when making linking decisions\footnote{Such an assumption is common within the networks literature to maintain tractability, see Jackson and Wolinsky (1996) or Watts (2001) for instance. Myopia is an appropriate assumption in financial networks where firm managers have myopic incentives. Such myopic incentives have been documented empirically in papers such as Jacobson (1993) and Mizik (2010). We relax this assumption in the extensions section where we allow for subsidies that change agent linking strategies.}. Each agent's utility is assumed to be linear in the benefits provided by each link and the linking cost. This also implies that agents are risk neutral and so consider the expectation over neighbor qualities when there is uncertainty. The flow utility of agent $i$ at any time $t$ is given by the following equation:

\begin{equation}
U_i = \sum\limits_{\{j\in K^t_i\}}(E[q_j|\mathcal{H}^t_i]-c)
\end{equation}

\subsection{Reputation and Learning Speed}
Since we have assumed a diffusion process, a sufficient statistic for all the individual link benefits is the average benefit per link produced by agent $i$ up to time $t$, which we denote as $B_i(t)$. Given our above assumptions, $B_i(t)$ follows a diffusion $dB_i(t) = q_i dt + (k^t_i \tau_i)^{-1/2} d Z_i(t)$ where the drift rate is the true quality $q_i$, the instantaneous volatility rate $(k^t_i \tau_i)^{-1/2}$ depends on the number of links agent $i$ has at time $t$, and $Z_i(t)$ is the standard Brownian motion with zero-drift and unit-variance. Importantly, this equation shows that the more links an agent has, the lower its volatility rate and the faster its true quality $q_i$ is learned. This is because an agent with more links produces more individual benefits, and so the average over all benefits is more precise. Note also that an agent with no links would not send any information, and thus there would be no learning about its quality. Therefore the topology of the network strongly affects the rate of learning about each agent's quality.

If at time $t$ all links of agent $i$ are severed, then no benefit will be produced by agent $i$ and this will be denoted as $b^t_i = \emptyset$. In this case no information is added and hence, the diffusion of agent $i$ is stopped at its current level. As mentioned, there is a prior belief of an agent $i$'s quality $\mathcal{N}(\mu_i, \sigma^2_i)$, and agents will update this belief in a Bayesian fashion in light of the observations of flow benefits. These observations combined with the prior quality distribution will result in a posterior belief distribution of agent $i$'s quality $f(q_i|\mathcal{H}^t_i)$ which is also normally distributed\footnote{As mentioned a sufficient statistic for the entire history is $B_i(t)$, so a neighbor only needs to know $B_i(t)$ in order to calculate this posterior.}. We denote $\mu^t_i = E[q_i|\mathcal{H}^t_i]$ as the expected quality of agent $i$ given the history $\mathcal{H}^t_i$ and call it the \textit{reputation} of agent $i$ at time $t$. The reputation represents the expected flow benefit of linking with agent $i$ at time $t$.

We have assumed that agents are myopic. Therefore, to maximize flow utilities, agent $i$ will cut off its link with agent $j$ once agent $j$'s reputation $\mu^t_j$ falls below the linking cost $c$. Since we assume all agents have homogeneous linking costs, and all neighbors have the same beliefs, any other agent that is linked to $j$ will also decide to sever its link. From this moment on, agent $j$ is effectively ostracized from the network; since it no longer has any links it cannot send any further information that could potentially improve its reputation\footnote{Although ostracism may seem harsh, as we noted earlier ostracism is a prevelant phenomenon that has been widely studied in the social psychology literature, in settings ranging from online interactions to office workplaces. Furthermore, in financial networks low reputation institutions may get shut out of funding completely during liquidity crisis.}. While in the base model an ostracized agent cannot return to the network, we relax this assumption in the extensions section.

\section{Network Dynamics and Stability}
\subsection{Network Dynamics}
The dynamics of the model evolve as follows: all pairs of agents that are neighbors according to the network constraint $\Omega$ will choose to link at time zero, since we have assumed that all agents have initial reputations higher than the cost $c$ (any agent with an initial reputation lower than $c$ is immediately ostracized from the network and would not need to be considered). Therefore the initial network at time 0 will be the same as the network constraint, $G^0=\Omega$. Over time agents that send bad signals will have their reputations decrease, and once an agent's reputation hits $c$ its neighbors will no longer wish to link with it. All its neighbors will sever their links and the agent is effectively ostracized from the network. We will show that this always happens for an agent with true quality $q_i\le c$, and will still happen with positive probability for an agent with quality $q_i>c$. The ostracization of an agent will affect its former neighbors as well. Since they now have once less link each, they will produce information about themselves more slowly than before, and so their reputations will be updated less quickly. 

The remaining agents in the network will continue to link and send signals until someone else's reputation drops too low and that agent is also ostracized. This process will continue until the qualities of all the remaining agents are known with very high precision and in the limit their reputations no longer change. Since agent qualities are fixed, by the law of large numbers any agents that remain in the network will have their qualities learned perfectly in the limit as $t\to \infty$, and the network will tend towards a limiting structure that we call the \textit{stable network}. The next section will explicitly characterize these stable networks, but we note that many different stable networks could potentially emerge depending on the true qualities of the agents and the signals they produce. Figure \ref{fig:figure1} shows the different network dynamics that could emerge even if the initial reputations of the agents are fixed, due to the uncertainty about the true qualities of the agents as well as the randomness in the signals they send.

\begin{figure}
\centering
\includegraphics[width=1\linewidth]{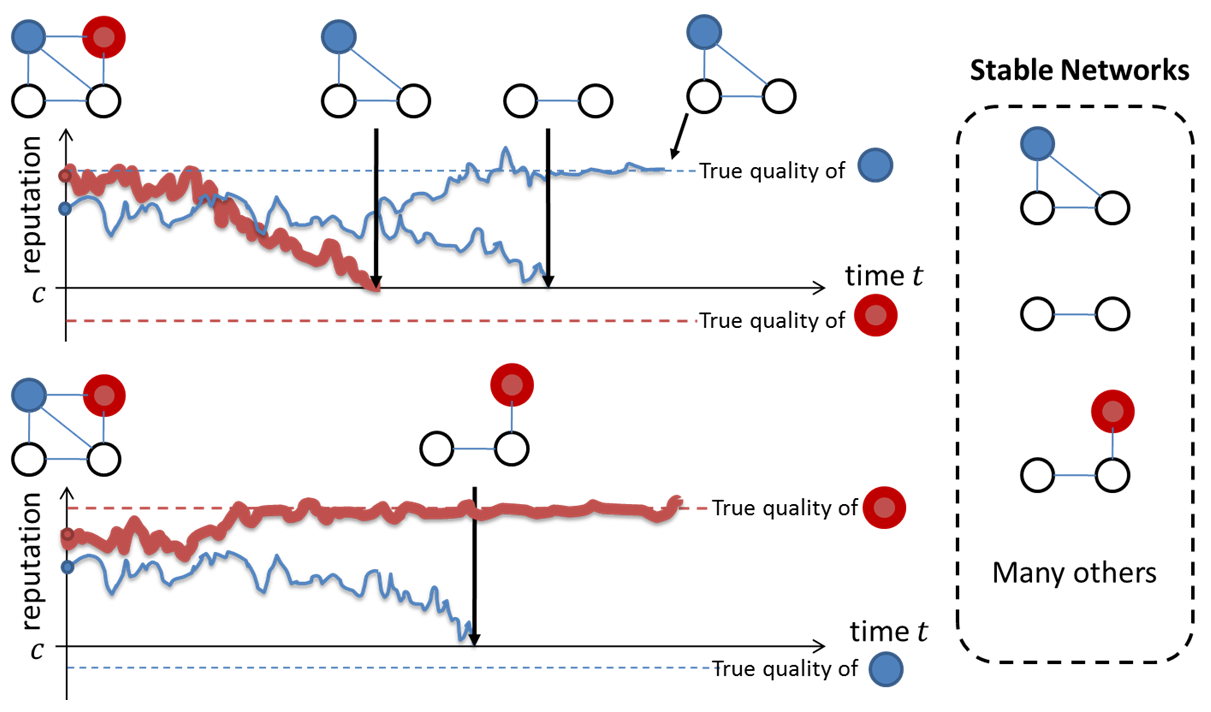}
\caption{Illustration of Possible Network Dynamics: From the same initial reputations for the red and blue agents, many different network dynamics and stable networks are possible. In the top graph the red (larger circle and bolded line) agent has a true quality less than $c$ and so will be ostracized from the network for certain at some time, while the blue (smaller circle and thin line) agent has a true quality above $c$ and so may or may not be ostracized from the network depending on the signals it sends. Each event leads to a different stable network, either with and one without the blue agent. In the bottom graph it is the blue agent who has a true quality lower than $c$ and so will be ostracized for sure, whereas the red agent could potentially stay in the network indefinitely.}
\label{fig:figure1}
\end{figure}

\subsection{Stable Networks}
As mentioned, we call the limiting network structure as $t$ goes to infinity, denoted by $G^\infty$, a stable network. Formally, let $G^\infty\equiv\lim_{t\to\infty}G^t$. This limiting structure always exists since agent qualities are fixed, so by the law of large numbers any agent that remains in the network will have its quality learned to an arbitrary precision over time. The probability that an agent who is still in the network at time $t$ ever becomes ostracized must therefore tend to zero as $t\to \infty$ (we show this analytically below). Which specific stable network eventually emerges is random and depends on the signal realizations of each agent. The tractability of our model allows us to explicitly characterize the set of stable networks that could emerge given a set of agents and a network constraint $\Omega$, as well as the impact of the rate of learning on the probability distribution over stable networks.

To understand which stable networks $G^\infty$ can emerge, we investigate whether a link $l_{ij}$ between agents $i,j$ can exist at $t = \infty$. If two agents $i$ and $j$ are not neighbors (i.e. $\omega_{ij} = 0$), then it is certain that $g^\infty_{ij} = 0$. If two agents $i$ and $j$ are neighbors (i.e. $\omega_{ij} = 1$), then the existence of this link $l_{ij}$ at $t = \infty$ requires that the reputations of both $i$ and $j$ never hit $c$ for all finite $t$, which means that neither agent is ever ostracized. Hence $G^\infty$ will always be a subset of the initial network $G^0$, and is composed only of agents whose reputations never hit $c$ for all finite $t$.

We say that an agent is \textit{included in the stable network} if their reputation never hits $c$ for all $t$, so that they are never ostracized from the network. \footnote{As a technical note, when we make the ostracization classification, we assume that an agent who has all its neighbors ostracized continues to send information about itself at its signal precision level, with the signals sent via the same probability distribution which is based on its true quality. So we still considered the agent ``ostracized" if its reputation drops to $c$ via this information process even after all its  neighbors have been ostracized. This assumption is made for technical purposes only and has no impact on the dynamics or the welfare of the model, as the agent has no links in this case.}

Note that being included in the stable network does not imply that an agent has any links in the stable network, as it could also be that all of the neighbors of that agent were ostracized even though the agent itself was not. We can calculate the \textit{ex ante} probability that an agent $i$ is included in the stable network, which we denote by $P(S_i)$ with $S_i$ denoting the event in which agent $i$ is included in the stable network. This can be accomplished using standard results regarding Brownian motion hitting probabilities, since $P(S_i)$ is equal to the probability that the agent's reputation never hits $c$ for all finite $t$. The following proposition gives this probability.

\begin{proposition} $P(S_i)$ depends only on the initial quality distribution and the link cost and can be computed by
\begin{equation}
P(S_i) = \int_{c}^\infty(1 - \exp(-\frac{2}{\sigma^2_i}(\mu_i - c)(q_i - c)))\phi\left((q_i - \mu_i)\frac{1}{\sigma_i}\right)dq_i/\sigma_i
\end{equation}
\label{prop:dist}
\end{proposition}
\begin{proof}
See appendix.
\end{proof}

Proposition \ref{prop:dist} has several important implications. Note that since $P(S_i)$ is positive and less than $1$ for all $i$, no agent is certain to be included in or excluded from the stable network.  Also note that the probability an agent is part of the stable network is independent of that agent's signal precision $\tau_i$. Therefore the rate at which the agent sends information does not affect the chance that it is in the stable network. This is because the rate at which the agent sends information only affects \textit{when} it gets ostracized from the network, but not \textit{if} it gets ostracized overall\footnote{To understand this intuitively, recall that reputation evolves through Bayes updating of the Brownian motion. A higher precision increases the amount of information sent at every moment in time, but the overall probability distribution of the information that is sent across all time remains the same. To see this rigorously, note that in the proof of Proposition \ref{prop:dist} in the appendix, the survival probability of an agent depends on $\tau _i$ only through the term $
t\tau_i$. Therefore increasing $\tau_i$ and decreasing the considered time $t$ proportionally leaves the overall survival probability unchanged.}. Furthermore, note that the probability an agent $i$ is included in the stable network is independent of its links with other agents and the properties of those agents. Connections with other agents affect the rate at which an agent sends information but not the agent's true quality, and so will not impact whether it is eventually ostracized from the network.

Using the explicit expression above, we can also describe how $P(S_i)$ depends on an agent's initial mean and variance, $\mu_i$ and $\sigma_i$.
\begin{corollary}
\label{corr:limitprob}
For each agent $i$, $P(S_i)$ is increasing in the mean of its initial quality $\mu_i$, decreasing in the variance of its initial quality $\sigma_i^2$, and decreasing in the link cost $c$. Moreover, $\lim_{\mu_i\to \infty} P(S_i) = 1$, $\lim_{\sigma_i\to 0} P(S_i) = 1$, $\lim_{c\to -\infty} P(S_i) = 1$.
\end{corollary}
\begin{proof}
See appendix.
\end{proof}
These properties are intuitive since an agent with a higher mean quality and smaller variance is less likely to have its reputation drop below $c$, and so is less likely to become ostracized. Moreover, lowering the linking cost also reduces the hitting probability since the agent's reputation would now have to fall lower to be excluded from the network.

As mentioned, $G^\infty$ must be a subset of $G^0$. Further, it can contain links only amongst pairs of agents that are both included in the stable network and were linked in the initial network. Equivalently, the set of stable networks can be thought of as the set of networks that can be reached from $G^0$ by sequentially ostracizing agents. Let $I\{S_i\}$ denote the indicator variable of the event in which agent $i$ is included in the stable network. Formally, a network can be stable if and only if it is a matrix with entries given by $g_{ij} = I\{S_i\}I\{S_j\}I\{g_{ij}^0=1\}$, for some realization of $\{S_i, ^\neg S_i\}_{i\in V}$. Links can exist only among agents that were never ostracized and were linked in the original network. Note that different realizations of $\{S_i, ^\neg S_i\}_{i\in V}$ could potentially correspond to the same stable network\footnote{For instance suppose that the network comprises only two agents $i$ and $j$. Then the event in which $S_i$ but not $S_j$ occurs and the event in which both $S_i$ and $S_j$ occur lead to the same stable network structure: the empty network.}.

By Proposition \ref{prop:dist}, we know that the rates of learning do not affect the probability of each event $S_i$. Since the rate of learning has no effect at an individual level, it cannot have an effect at the aggregate level either. This is formalized in the following theorem. We can also use the equation in Proposition \ref{prop:dist} to derive an analytic expression for the probability that any specific stable network emerges, which is presented in the corollary below. Figure \ref{fig:compute_prob} in the appendix shows an example of how the corollary can be applied to a simple network of three agents.

\begin{theorem}
\label{sigprec}
The signal precisions of the agents, $\{\tau_i\}_{i\in V}$, do not affect the set of stable networks that can emerge or the probability that any stable network emerges.
\end{theorem}
\begin{proof}
It is clear that a network $G$ must be a subset of $G^0$ and can be stable if and only if there exists at least one combination of events $\{S_i, ^\neg S_i\}_{i\in V}$ such that $g_{ij} = I\{S_i\}I\{S_j\}I\{g_{ij}^0=1\}$. Thus the set of stable networks does not depend on the learning speed. Moreover, according to Proposition \ref{prop:dist}, $P(S_i)$ is independent over the different agents and does not depend on the speed of learning. Hence the probability that any specific link exists in the stable network exists also independent of the learning speed, so the probability of any stable network emerging is also independent of the learning speed.
\end{proof}

\begin{corollary}
\label{corr:stableprob}
The probability that a network $G$ is a stable network is given by $\sum\limits_{\{S_i\}}\prod_i P(S_i)$ where the summation is over all realizations of $\{S_i, ^\neg S_i\}_{i\in V}$ that correspond to $G$.
\end{corollary}

We have shown that the speed of learning has no impact on the probability that a network $G$ is stable. This is intuitive since learning only affects the duration of a link but not its final state. However, learning will have a crucial role on the social welfare of a network, which directly depends on how long the agents are connected. We will consider the impact of learning on the social welfare in the next section.

\section{Welfare Computation}
We will analyze overall social welfare from an \textit{ex ante} perspective, given only the network constraint $\Omega$ and the prior agent quality distributions. Importantly the \textit{ex ante} welfare is calculated before the agent qualities are learned and any signals are sent. This type of welfare is the most suitable for the type of design settings we will consider later, as it requires the least knowledge on the part of the network designer. Let $P(L^t_{ij}|q, G^0)$ denote the probability that the link between agents $i$ and $j$ still exists at time $t$. Also, let the parameter $\rho$ represent the discount rate of the network designer\footnote{We are assuming that the designer itself is more patient than the myopic agents. This can be thought of, for instance, as a company manager who is more patient than its workers who act myopically in their interactions, or a financial regulator that is more patient than the financial institutions, which have managers with myopic incentives.}. We can define the overall \textit{ex ante} social welfare $W$ formally as follows:

\begin{equation}
\label{welfare0}
W = \int_{q_{1}=-\infty}^\infty ... \int_{q_{N}=-\infty}^\infty \sum\limits_{i,j}\int_{0}^\infty e^{-\rho t}(q_j - c)P(L_{ij}^t|q,G^0)dt\phi(\frac{q_N - \mu_N}{\sigma_N})dq_N/\sigma_N...\phi(\frac{q_1 - \mu_1}{\sigma_1})dq_1/\sigma_1
\end{equation}

We will show that this social welfare expression can be calculated in a tractable fashion using a somewhat indirect approach. This approach utilizes the fact that the \textit{ex ante} social welfare is an expectation over all the possible \textit{ex post} signal realizations. We can calculate the \textit{ex ante} welfare by integrating over all possible realizations of the \textit{ex post} welfare, which simplifies equation \ref{welfare0} to a much more tractable form.

\subsection{Ex post welfare}
Consider an \textit{ex post} realization of agent hitting times $\varepsilon = \{\varepsilon_i^{t_i}\}_{i\in\mathcal{V}}$, where $\varepsilon_i^{t_i}$ denotes the event in which agent $i$'s reputation hits $c$ at time $t_i$ given all the agent signals (note that $t_i = \infty$ means that agent $i$'s reputation never hits $c$). In the event in which $t_i < \infty$, since the belief at time $t_i$ is correct, the expected value of agent $i$'s quality conditional on this event $\varepsilon_i^{t_i}$ is $E[q_i|\varepsilon_i^{t_i < \infty}] = c$. In the event with $t_i = \infty$, since the initial belief is accurate in expectation
\begin{eqnarray}
\mu_i = E[q_i] = P(\varepsilon_i^{t_i < \infty}) E[q_i|\varepsilon_i^{t_i < \infty}] + P(\varepsilon_i^{t_i = \infty}) E[q_i|\varepsilon_i^{t_i = \infty}] \\
=(1-P(S_i))c + P(S_i)E[q_i|\varepsilon_i^{t_i = \infty}]
\end{eqnarray}
and we have
\begin{equation}
E[q_i|\varepsilon_i^{t_i = \infty}] = \frac{\mu_i - (1-P(S_i))c}{P(S_i)}
\end{equation}
where $P(S_i)$ is given by Proposition \ref{prop:dist} and is independent of the network and the learning speed.

According to the above discussion, given an \textit{ex post} realization $\varepsilon$, an agent $i$ obtains 0 surplus from its neighbors that have finite hitting times and obtains positive surplus from those whose reputation never hits $c$ (and are therefore included in the stable network). The exact benefit agent $i$ receives in the second case depends on its own hitting time $t_i$, which determines the link breaking time with the other agents. We can calculate the \textit{ex post} surplus that an agent $i$ receives given $\varepsilon$ as follows:
\begin{eqnarray}
W_i(\varepsilon) = E_{q|\varepsilon} \left[\sum\limits_{j:g^{0}_{ij} = 1}\int_{0}^{\min\{t_i,t_j\}}e^{-\rho t}(q_j - c) dt\right]\\
=\sum\limits_{j:g^{0}_{ij} = 1, t_j = \infty}\int_{0}^{t_i} e^{-\rho t} \left(\frac{\mu_j - (1-P(S_j))c}{P(S_j)} -c\right) dt\\
=\frac{1-e^{-\rho t_i}}{\rho} \sum\limits_{j:g^{0}_{ij} = 1, t_j = \infty}\frac{\mu_j-c}{P(S_j)}
\end{eqnarray}
Note that this $W_i$ is taken from the perspective of the designer as it incorporates futures payoffs at the discount rate of $\rho$. This equation shows that in each \textit{ex post} realization of other agent hitting times, agent $i$ benefits if $t_i$ increases and it is ostracized later from the network. Summing over all agents, the social welfare given the \textit{ex post} realization $\varepsilon$ is therefore

\begin{equation}
\label{welfare}
W(\varepsilon) = \sum\limits_{i}\left(\frac{1-e^{-\rho t_i}}{\rho} \sum\limits_{j:g^{0}_{ij} = 1, t_j = \infty}\frac{\mu_j-c}{P(S_j)}\right)
\end{equation}

By taking the expectation over the events $\varepsilon$, the \textit{ex ante} social welfare can be found as $W = E_{\varepsilon}[W(\varepsilon)]$. In order to compute the \textit{ex ante} social welfare, we still need to know the distribution of the $t_i$, which is coupled in a complicated manner with the initial network and the learning process. For instance, if the neighbor of agent $i$ has a low hitting time and is ostracized quickly, then agent $i$ sends information at a slower rate and its own hitting time would increase. Thus directly computing the social welfare using the above equation is still difficult. In the next subsection, we develop an indirect method to calculate the distribution of $t_i$.

\subsection{Hitting time mapping}
Recall that an agent's links will scale up the rate at which it sends information compared to the rate it would send information if its precision were constant at the base level of $\tau_i$. Therefore each link also scales down the time at which the agent's reputation hits $c$. So to calculate when the agent is ostracized, we can first find when the agent's reputation would hit $c$ through sending signals at its signal precision level, and then scale this time downwards proportionately based on the network effect\footnote{Refer to footnote 12 for a justification of this type of scaling.}. Consider an \textit{ex post} realization of hitting times $\hat\varepsilon = \{\hat\varepsilon_i^{t_i}\}_{i\in\mathcal{V}}$ in which agent $i$'s reputation would hit $c$ at time $t_i$ if its precision were fixed at $\tau_i$ at all times. Note that the events $\hat\varepsilon_i^{t_i}$ are independent from each other across different agents, and since the precision is fixed they also do not depend on the network structure. The probability of $\hat\varepsilon_i^{t_i}$ can be explicitly computed in the following lemma.
\begin{lemma}
\label{pdf}
The probability density function $f(\hat\varepsilon_i^{t_i}), \forall t_i < \infty$ can be computed as
\begin{equation}
f(\hat\varepsilon_i^{t_i}) = \int_{-\infty}^\infty  \frac{\mu_i-c}{\sigma_i^2\sqrt{\tau_i}}t^{-3/2}\phi\left(\sqrt{t\tau_i}(q_i-c) + \frac{\frac{1}{\sigma^2_i}(\mu_i - c)}{\sqrt{t\tau_i}}\right) \phi(\frac{q_i-\mu_i}{\sigma}) dq_i/\sigma_i
\end{equation}
The probability mass point function $f(\hat\varepsilon_i^{t_i = \infty}) = P(S_i)$.
\end{lemma}
\begin{proof}
See appendix.
\end{proof}
Using Lemma \ref{pdf}, we can easily obtain the distribution of joint events $f(\hat\varepsilon)=\prod_i f(\hat\varepsilon_i^{t_i})$ due to the fact that the individual events are independent. This would measure the joint probability of the agents exiting the network at times $\{t_i\}_{i\in\mathcal{V}}$ if the information sending speed of the agents were not being scaled by the number of their links. If there were no network effect, the \textit{ex ante} social welfare could be directly computed using the distribution of hitting times given by Lemma \ref{pdf}. However, due to the network effect, the actual hitting times may vary for each $\hat\varepsilon$. We can define $M:[0,\infty]^N \to [0, \infty]^N$ be the hitting time mapping function, which maps the hitting times with no network effect to the actual hitting times when there is a network effect. In the appendix we present an algorithm for computing $M$, which operates by scaling the information speed of each agent at every time $t$ by their current number of neighbors and updating the speed at which an agent sends information when a neighbor is ostracized. Note that if $t_i = \infty$ in the event $\hat\varepsilon_i^{t_i}$ then it is also $\infty$ in the mapped event $\varepsilon_i^{t_i}$. This means that an agent that never leaves the network with no scaling effect will not leave when the times are scaled either. Then given a realization $\hat\varepsilon$, the \textit{ex post} social surplus can be computed as
\begin{equation}
W(\hat\varepsilon) = \sum\limits_{i}\left(\frac{1-e^{-\rho M_i(t)}}{\rho} \sum\limits_{j:g^{0}_{ij} = 1, t_j = \infty}\frac{\mu_j-c}{P(S_j)}\right)
\end{equation}

Therefore, the \textit{ex ante} social welfare is $W = E_{\hat{\varepsilon}}[W(\hat\varepsilon)]$. We note that this is a tractable equation for the \textit{ex ante} social welfare given any network structure and set of agents. Proposition \ref{prop:dist} gives the explicit expression for $P(S_j)$, and Lemma \ref{pdf} provides the distribution of $\hat\varepsilon$. Thus our model allows for easy and tractable computations of the \textit{ex ante} social welfare of any type of network. Theorem \ref{welfcalc} below formalizes this result.

\begin{theorem}
\label{welfcalc}
Given $\Omega$, the initial quality distributions, and the link cost $c$, the overall \textit{ex ante} social welfare can be computed as follows
\begin{equation}
W = E_{\hat\varepsilon}\left[ \sum\limits_{i}\left(\frac{1-e^{-\rho M_i(t)}}{\rho} \sum\limits_{j:g^{0}_{ij} = 1, t_j = \infty}\frac{\mu_j-c}{P(S_j)}\right)\right]
\end{equation}
where the distribution of $\hat\varepsilon$ is computed using Lemma \ref{pdf} and the hitting time mapping function $M$ is given in the appendix.
\end{theorem}

\section{Impact of Information and Learning}
In this section we study the impact of learning on \textit{ex ante} welfare, both individual and overall, given an initial network $G^0$. In particular, we will show how the agents' signal precisions, a representation of the rate of learning, impact individual agent welfare as well as the overall social welfare.

As a benchmark, we consider the social welfare when there is no learning, which we denote by $W^*$. When there is no learning, no existing link will be severed. The social welfare of an agent $i$ without learning can therefore be computed by summing over the mean qualities of all agents it is connected with initially:

\begin{equation}
W^*_i =  \sum\limits_{j:g^0_{ij} = 1} \int_0^\infty e^{-\rho t} (\mu_j - c) dt
= \frac{1}{\rho}\sum\limits_{j:g^0_{ij} = 1} (\mu_j - c)
\end{equation}

The \textit{ex ante} overall social welfare without learning is given by the sum over the individual welfares:
\begin{equation}
W^* = \sum\limits_i W^*_i
= \frac{1}{\rho}\sum\limits_i \sum\limits_{j:g^0_{ij} = 1} (\mu_j - c)
\end{equation}

\subsection{Overall Impact of Learning}
Let $W(\tau_1,...,\tau_N)$ be the \textit{ex ante} social welfare when agents learn each other's true quality with the signal precisions being $\tau_1,...,\tau_N$. We also let $W_i(\tau_1,...,\tau_N)$ represent an agent $i$'s \textit{ex ante} welfare given these signal precisions. The next theorem states that in any network, the addition of learning has a negative impact on every individual's \textit{ex ante} welfare for any value of the signal precisions. This immediately implies that it lowers the overall \textit{ex ante} social welfare as well.
\begin{theorem}
$W_i(\tau_1,...,\tau_N) < W_i^*$ for all $i$ and for all ${\tau _1,...\tau_N}$.
\label{compnl}
\end{theorem}
\begin{proof}
See appendix.
\end{proof}
There are two main factors that are at work in this result. First, the myopia of the agents causes the learning to be done inefficiently. Second, cutting off a link imposes a negative externality on the agent who is ostracized, since that agent can no longer receive benefits from its neighbors. Taken together, these factors lead to a reduction in overall social welfare. More precisely, when a link $l_{ij}$ is severed due to agent $j$'s reputation hitting $c$, agent $i$ does not gain welfare compared to the case without learning. This is because the expected value of having a link with $i$ from $t^*_j$ on is 0 and thus having the link or not makes no difference\footnote{Agent myopia is causing the cut-off value to be too high, and so the agent does not benefit from its learning. This feature of reputational learning is similar to that shown in van der Schaar and Zhang (2014). In Section VIII we discuss a possible solution for this problem by providing agents a subsidy to increase experimentation.}. However, agent $j$ loses welfare compared to the case without learning because agent $i$'s reputation is still above the link cost and thus having the link would benefit $j$ over not having the link. 

This result supports the damaging impacts of ostracism found in the social psychology literature, which were mentioned above in the literature review. The social psychology literature usually documents the harmful effects of ostracism from the perspective of the agents that have become ostracized and can no longer benefit from interactions with the other agents. However, our result goes further by stating that the possibility of ostracism will actually lower \textit{every} agent's social welfare from an \textit{ex ante} perspective. By allowing for the ostracism of others, agents open themselves up to ostracism as well, which lowers their own welfare by more than they benefit from ostracizing other agents.  Theorem \ref{compnl} shows that every agent is hurt \textit{ex ante} by ostracism, even those that wouldn't themselves be ostracized in the majority of the \textit{ex post} realizations of the network.

\subsection{Impact of Individual Information}
The previous result showed that learning is harmful on aggregate: under learning both individual and overall network welfare are lower than without learning. However, we show in this subsection that learning need not be harmful at an individual level, as the rate that a single agent sends information changes. We now investigate more closely how the information generation rate of a single agent (i.e. an agent's signal precision) affects welfare. The faster an agent generates information about its own reputation, the faster the other agents will learn its true quality (if the link is not broken).

First we characterize the impact of an agent's signal precision on that agent's own welfare. The next proposition shows that sending more information about itself will always harm an agent.
\begin{proposition}
$W_i(\tau_i, \tau_{-i})$ is strictly decreasing in $\tau_i$.
\label{prop:welfsig}
\end{proposition}
\begin{proof}
Consider any \textit{ex post} realization $\varepsilon = \{\varepsilon_i^{t_i}\}_{i\in V}$. If $t_i = \infty$, then changing $\tau_i$ alone does not change the fact that agent $i$ would stay in the network forever, as so it does not affect the hitting time realization of any other agent either. Therefore agent $i$'s welfare $W_i(\varepsilon)$ is not affected. If $t_i < \infty$, then the welfare of agent $i$ depends on (1) the expected quality of all the neighboring agents $j$ whose $t_j = \infty$ and (2) its own hitting time $t_i$. Since (1) is not affected by changing $\tau_i$, we only need to study how $\tau_i$ affects $t_i$.

Intuitively $t_i$ is decreasing in $\tau_i$ since agent $i$'s information sending speed is faster due to a higher precision. We provide a more rigorous proof by contradiction as follows. Suppose agent $i$'s new hitting time increases to $t_i' = t_i + \Delta > t_i$. In this new realization, consider the duration from $0$ to $t_i$. Since $t'_i > t_i$, all other agents' information sending process and speed do not change before $t_i$. Hence, agent $i$'s instantaneous precision at $t \leq t_i$ changes to $(\tau^t_i)' = \frac{\tau'_i}{\tau_i}\tau^t_i$. Hence, information sending by agent $i$ is faster at any moment in time before $t_i$. Since, the stopping time $t'_i$ is larger than $t_i$, the total amount of information sent by agent $i$ given $\tau'_i$ is larger than that given $\tau_i$. Because the total information sent should remain the same, this causes a contradiction. Therefore $t_i'$ should be smaller than $t_i$ for a larger $\tau_i'$.
\end{proof}

This result is in accordance with Theorem \ref{compnl} and shows that an agent sending information about itself will strictly decrease its own welfare. This is because in each realization in which the agent is ostracized from the network, the agent will now be ostracized sooner and hence it will enjoy less benefits from others. Since the agent already starts out with the maximal amount of links it can obtain, it in effect has nothing to gain and everything to lose by allowing its own reputation to vary. We relax this assumption in the extensions section and allow agents to form new links with those they are not connected with initially; under those circumstances an agent will be able to benefit by generating more information about itself.

Though increasing the information sending speed is always harmful for an agent itself, it can actually be helpful to its direct neighbors. The next proposition provides a sufficient condition on the initial network such that this holds.

\begin{proposition}
Given an initial network $G^0$, for any two initially connected agents $i$ and $j$ that are linked through a unique path (i.e. the direct link), increasing one's precision increases the other's welfare.
\label{prop:linkwelf}
\end{proposition}
\begin{proof}
Consider any \textit{ex post} realization $\varepsilon = \{\varepsilon_i^{t_i}\}_{i\in V}$. If $t_i = \infty$, then increasing agent $i$'s signal precision $\tau_i$ does not change the realization $\varepsilon_i^{t_i}$. Hence $t_j$ is not affected. If $t_i < \infty$, then according to Proposition \ref{prop:dist}, the new hitting time $t'_i$ is sooner if agent $i$'s signal precision is larger. This causes the link between agent $i$ and $j$ to be severed (weakly) sooner, leading to a (weakly) later hitting time of agent $j$ because agent $j$ will send information at a slower speed for a longer time. Since changing agent $i$'s signal precision does not change the finiteness of the hitting time of all other agents, agent $j$'s welfare increases due to a longer hitting time for itself.
\end{proof}

Since the information sending speed of agent $j$ slows after agent $i$ is ostracized, agent $j$'s hitting time is larger. Agent $j$ therefore prefers its direct neighbor to send more information, so that it can cut off more quickly in case the neighbor is bad. After the link is broken, agent $j$ will also be able to reveal less information about itself, which is beneficial according to Proposition \ref{prop:welfsig}. In this way agent $j$ would enjoy more benefits for a longer time from its links with its other neighbors. We can extend this analysis for more distant agents when the two agents are connected through a unique path. This is summarized in the corollary to Proposition \ref{prop:linkwelf} below.

\begin{corollary}
\label{corr:uniquepath}
Given any initial network $G^0$, for any two agents $i$ and $j$ that have a unique path between them, increasing one's signal precision decreases/increases the other's welfare if they are an odd/even number of hops away from each other.
\end{corollary}

The above result shows an odd-even effect of the distance between two agents on the agent's welfare. In all minimally connected networks (such as star, tree, forest networks), any two agents have a unique path between each other and so the impact of any agent's information sending speed on any other agent's welfare can be completely characterized. 

\begin{figure}
\centering
\includegraphics[scale=.7]{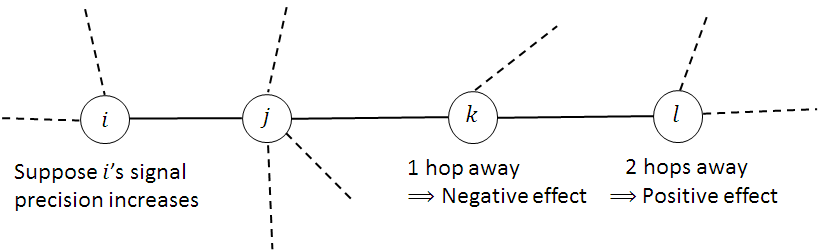}
\caption{Example for Corollary \ref{corr:uniquepath}}
\label{fig:line}
\end{figure}

As an example, consider a network where four agents ${i,j,k,l}$ are connected via a unique path, as depicted in Figure \ref{fig:line}. Agent $i$ is linked with agent $j$, agent $j$ is linked with agent $k$, and agent $k$ is linked with agent $l$. Then if agent $i$ sends more information about itself, it stays connected with agent $j$ for a shorter period of time. This causes agent $j$ to send less information about itself, causing agent $k$ to cut off its link with $j$ more slowly if $j$ were to be ostracized. Then agent $k$ is able to link with its other neighbors for a shorter length of time in expectation, decreasing the \textit{ex ante} welfare of $k$. Therefore agent $k$ is hurt when the neighbor of its neighbor, agent $i$, sends more information. However agent $l$ now links with its own neighbors for a longer length of time, and so it benefits when $i$ sends more information. However, when there are multiple paths between agents, which implies there are cycles in the network, the impact of the signal precision of an agent on the other agents' welfares is much less clear. The reason is that with cycles the neighbor of an agent $i$'s neighbor may also be linked with agent $i$ itself\footnote{This is known in the social network literature as triadic closure.}, and so the positive and negative effects of information from Corollary \ref{corr:uniquepath} are entangled together. The following proposition shows that even for an immediate neighbor, the impact could be totally opposite of Proposition \ref{prop:linkwelf} when cycles are present in the network.

\begin{proposition}
If the initial network $G^0$ has cycles, then it is possible that increasing some agent's signal precision decreases its immediate neighbor's welfare.
\label{prop:cycle}
\end{proposition}
\begin{proof}
We prove by constructing a counterexample, which is shown in Figure \ref{fig:relay}. Consider a network with $K > 3$ agents. Agents 1, 2, 3 form a line and the other $K-3$ agents connect to both and only agents 1 and 2. We assume that agent 3's true quality is perfectly known (initial variance 0) and large. Hence, agent 3's reputation never hits $c$. We also assume that the mean qualities of agents 4 to $K$ are close to $c$. Hence, agent 2 almost does not gain benefit from those agents even when $K \to \infty$.

Consider a realization in which agent $1$'s reputation hits $c$ at $t_1 < \infty$ and agent $2$'s reputation hits $c$ at $t_2<\infty$. By increasing the signal precision of agent $1$, its hitting time decreases to $t'_1 < t_1$. If $t'_1 > t_2$, then agent 2's hitting time is not affected, i.e. $t'_2 = t_2$. Otherwise, the new hitting time may be different from $t_2$. To simplify the analysis, we consider the extreme case in which $\tau_i \to \infty$, thereby $t'_1 \to 0$. Therefore, agent 2 loses the link with agent 1 from the beginning in any realization. However, since agents 3 to $K$ also lose the link with agent 1 from the beginning, for those whose hitting time was earlier than $t_2$, their hitting time would increase by a factor of 2. If there are at least three agents among 4 to $K$ whose hitting was between $[t_2/4, t_2/2]$, agent 2's information sending speed will increase sufficiently much that agent 2's hitting time is smaller. By making $K$ large we can always making the probability of this event be large enough. Thus, agent 2's hitting time will decrease on average.
\end{proof}

\begin{figure}
\centering
\includegraphics[scale=.8]{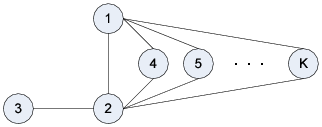}
\caption{Counterexample for Proposition \ref{prop:cycle}}
\label{fig:relay}
\end{figure}

We have seen that increasing the information sending speed of an individual agent $i$ could be both good or bad for other agents depending on their locations in the network and their relation with agent $i$. We note that it could similarly be good or bad for overall social welfare. So in contrast with Theorem \ref{compnl}, increasing the amount of information about a single agent can benefit the network overall. This would happen for instance, if there are three agents, $i$, $j$, and $k$ who are connected in a line, with links $ij$ and $jk$. Suppose that the mean of agent $k$'s quality is much higher than those of the other two agents. Then most of the welfare in this network comes through the link between agents $j$ and $k$. If agent $i$ sends more information, agent $j$ would be able to preserve its link with agent $k$ for a longer period of time, and overall social welfare would increase. This example highlights how critical the network structure is in determining the overall impact of more information by a single agent.

\section{Optimal Networks}
In this section, we study which underlying network constraints $\Omega$ maximize the overall \textit{ex ante} social welfare. Equivalently, we could think of a benevolent network planner that wishes to maximize social welfare by designing the network constraint $\Omega$ through designating which agents are able to form links with which other agents. For instance, in the financial network setting we could think of a regulator that specifies which types of financial institutions are allowed to transact with which other types of institutions in order to maximize overall social welfare\footnote{We note that many other types of objection functions are also possible instead of the overall \textit{ex ante} social welfare. For instance the designer may wish to maximize network welfare generated over a certain time interval, or before a set deadline is reached. Or the designer may weigh the welfare of some agents more heavily than that of others. Given the tractability of our model, many of our results can be extended for these alternative settings.}. 

\subsection{Fully connected networks}
One intuition is that a fully connected network, with no constraints on links, would be optimal since it results in the largest number of links initially, and we have assumed that all agents have an initial reputation higher than the linking cost $c$. This intuition is accurate in certain cases, such as if the designer is extremely impatient (i.e. $\rho \to \infty$). Since the designer cares only about the initial time period, and when time is short almost no new information can be learned, it is best to design the network based on the agents' starting reputations. Surprisingly though, the fully connected network is also optimal on the other extreme, when the designer is completely patient (i.e. $\rho \to 0$). In this case, the designer cares about the social welfare of the stable network that eventually develops, and allowing all agents to be connected initially leads to the largest probability of links in the final stable network. We prove these welfare results in the following proposition.

Further, note that the designer's level of patience is inversely related with the rate of learning, as faster learning means that information is revealed sooner and thus less patience is required. Therefore a similar result holds for the rate of learning: as the rate of learning becomes extremal the fully connected network becomes optimal as well.  So for instance, a financial regulator should optimally let all types of financial institutions transact with each other if it is very patient or very impatient, or the information production is extremely fast or slow.

\begin{proposition}
\label{limitlearn}
1. If the designer is either completely impatient (i.e. $\rho \to \infty$) or completely patient (i.e. $\rho \to 0$), the optimal $\Omega$ is the fully connected network. 
\newline
2. Fix the other parameters of the model and suppose the agents' signal precisions are all multiplied by the same constant $\lambda$. If learning becomes very fast (i.e. $\lambda \to \infty$) or very slow (i.e. $\lambda \to 0$), then the optimal $\Omega$ is the fully connected network.
\end{proposition}
\begin{proof}
See appendix.
\end{proof}

When the designer is either completely  patient or impatient, the social welfare depends only on the network $G^0$ or $G^\infty$, respectively. The exact hitting time does not affect the social welfare. Similarly if the learning is very slow, then the network structure always remains at $G^0$, and if the learning is very fast then $G^\infty$ is realized very quickly, so in both cases a fully connected network is optimal. The idea is that in both extremes, the exact path of learning is no longer critical and so the negative externalities of information are mitigated.

For intermediate levels of patience or learning however, changes in individual agent hitting times due to linking could have a significant impact on the social welfare. We will show later that having all agents fully connected with each other is not always the optimal choice. In the next proposition though we show that the fully connected network is still optimal in the case where the agents are homogeneous and have very high initial qualities.
\begin{proposition}
Suppose all agents are \textit{ex ante} identical. Fixing the other parameters, there exists $\bar{\mu}$ such that if $\mu_i > \bar{\mu}$ $\forall i$, then the optimal $\Omega$ is the fully connected network.
\label{prop:compnet}
\end{proposition}
\begin{proof}
We will prove that for  $\bar{\mu}$ large enough, the social welfare of any non fully connected network will be increased through the addition of any new link. Therefore the welfare of the fully connected network will be greater than the welfare of any other network. Consider an arbitrary network constraint $\Omega$ that is not fully connected. Suppose that a link between agents $i$ and $j$ is added to the network, and consider the welfare of the new network constraint $\Omega'$. 

First consider the change in welfare of agent $i$. In any realization where agent $i$ is ostracized, its welfare through having the extra link with $j$ decreases by no more than $\frac{(N-2)\mu}{\rho}$, the welfare loss when it loses all its links with the other agents immediately. In any realization where agent $i$ is not ostracized, its welfare with the additional link increases by $\frac{\mu}{\rho}$, the discounted value of the new link given the expected quality of agent $j$. Thus the change in welfare for agent $i$ is bounded below by $P(S_i)\frac{\mu}{\rho}+(1-P(S_i))\frac{(N-2)\mu}{\rho}=\mu(P(S_i)\frac{(N-1)}{\rho}-\frac{(N-2)}{\rho})$. Similarly, we can show that the change in welfare for agent $j$ is bounded below by $\mu(P(S_j)\frac{(N-1)}{\rho}-\frac{(N-2)}{\rho})$. 

Now consider the change in welfare for all the other agents in the network. In any realization where both agent $i$ and agent $j$ are not ostracized, the hitting times of all the agents in the network are unaffected by the new link. In any realization where either agent $i$ or agent $j$ are ostracized, the change in welfare for all the other agents is bounded below by $\frac{(N-2)(N-1)\mu}{\rho}$. Thus the total change in welfare for all other agents in the network is bounded below by $[P(S_i)(1-P(S_j))+(P(S_j)(1-P(S_i))+(1-P(S_i))(1-P(S_j))]\frac{(N-2)(N-1)\mu}{\rho}$.

Combining the above two observations, we note that the change in welfare for the whole network is bounded below by $\mu[P(S_i)\frac{(N-1)}{\rho}-\frac{(N-2)}{\rho}+P(S_j)\frac{(N-1)}{\rho}-\frac{(N-2)}{\rho}+P(S_i)(1-P(S_j))+(P(S_j)(1-P(S_i))+(1-P(S_i))(1-P(S_j))\frac{(N-2)(N-1)}{\rho}]$. When $\bar{\mu}$ is large, $P(S_i)$ converges to $1$ by Proposition \ref{prop:dist}. Thus for $\bar{\mu}$ large enough, the lower bound for the change in welfare of agents $i$ and $j$ converges to $\frac{2(N-1)\mu}{\rho}$, a positive number. 

When $\bar{\mu}$ is large, $P(S_i)$ and $P(S_j)$ converge to $1$ by Proposition \ref{prop:dist}. Therefore the lower bound for the change in welfare converges to $\frac{2\mu}{\rho}$, a positive. 
\end{proof}

\subsection{Core-periphery networks}
As agents become more heterogeneous in terms of their initial expected quality, it can be optimal to constrain connections among agents. Suppose agents are divided into two separate types, and the initial mean quality of the high type agent is $\mu_H$ while the initial mean quality of the low type agent is $\mu_L< \mu_H$. We show that when the expected qualities of the two types are sufficiently different, the optimal network constraint has a core-periphery structure\footnote{Although this theorem assumes there are exactly two types, a similar result holds if instead the agents are composed of two groups and within each group have parameters that are sufficiently close together.}.

\begin{theorem}
\label{coreperiph}
Suppose that there are two groups of agents, one with initial reputation $\mu_L$ and one with initial reputation $\mu_H$. Fixing all other parameters, there exists $\bar{\mu}$ such that $\forall \mu_H > \bar{\mu}$, the optimal $\Omega$ is a core-periphery network where all high type agents are connected with all other agents and no two low type agents are connected. ($\bar{\mu}$ will depend on the other network parameters.)
\end{theorem}
\begin{proof}
We first show that all high type agents should connect to all other high type agents. This is based on a similar argument as in the proof of Proposition \ref{prop:compnet}. Since when $\mu_H \to \infty$, all high type agents will stay in the stable network with very high probability, adding a link between any two high type agents will strictly improve their welfare while impacting the welfare of all other agents with very low probability. Hence, there must exist a large enough value for $\mu_H$ such that the welfare of high type agents is maximized when all high type agents connect to all other high type agents in the initial network.

Next we show that all low type agents should not connect to each other in any network where each is linked to at least $1$ high type agent. When $\mu_H \to \infty$, the welfare obtained by a link with any low type agent $j$ is dominated by that a link with high type agents, i.e. we can suppose that the welfare received by a link with another low type agent is approximately zero in comparison to a link with the high type agents. Having additional links with other low-type agents reduces the hitting time of agent $j$, $M_j(t)$, in the event that it gets ostracized, thereby reducing agent $j$'s welfare by more than the welfare gain of the additional link. Therefore, low type agents do not connect to each other in the optimal initial network. 

Finally we show that all low type agents should connect with every high type agent. Since the probability that the high type agent is ostracized approaches zero, such a link does not affect them relative to the extra welfare that the low type agents receive. Therefore we consider only the effect on the welfare of the low type agent to be connected with all high type agents. In a realization where the low type agent is not ostracized, this is optimal for all agents, as the high type agent stays in the network with very high probability when $\mu_H$ is large enough. Thus both agents have their welfare increased while not affecting the welfare of all other agents. We show that it is also optimal in realizations where the low type agent is ostracized. Again we will assume that the high type agent is not ostracized, which will hold for $\mu_H$ high enough. The low type agent receives a flow payoff of $\mu_H$ from every high type agent that it has an active link with. Note that in the hitting time mapping function the hitting time of an ostracized agent $i$ is scaled by $1/K$, where $K$ is the total number of high type neighbors. Thus the decrease in hitting time is exactly balanced out by the increase in flow payoff in the case without discounting, and with discounting it is strictly better for the low type agent to have an extra link.

\end{proof}

The above result shows that under the optimal network constraint, high reputation agents should be placed in the core and connected with all other agents, while low reputation agents should be placed in the periphery and not connected with other low reputation agents. Therefore agents with lower initial reputations should be placed in less central positions within the network in order to mitigate the negative effects of ostracism. Allowing low reputation agents to connect with too many other agents would increase the rate at which they send information, causing them to be ostracized sooner and hurting them more than they would gain through the direct benefits of the extra links. This core-periphery structure is commonly seen in many real-world financial networks, with large well capitalized banks in the core and smaller banks in the periphery. A reason for this could be that the greater reputation of large banks lets them withstand negative shocks more easily without being ostracized by their counterparties. Smaller banks produce less information through their lesser number of transactions, allowing them to avoid being ostracized as quickly.\footnote{We note that financial regulators have started imposing core-periphery structures on various financial networks to encourage stability. Many banks are now required to trade through a central clearing counterparty (CCP), which is a large financial institution that is ideally very stable. The idea is that trading with the CCP will mitigate the uncertainties that individual banks have about each other's qualities and thus help prevent liquidity runs during financial crisis.}

We note that the above result depends heavily on the type of learning environment that is present. From Proposition \ref{prop:compnet}, we know that if the designer was either very patient or impatient, or if learning was very slow or very fast (relative to the parameters of the agents), then the optimal initial network would be the fully connected network. Fixing the agent reputations, a core-periphery constraint structure is only optimal at intermediate levels of learning.

\subsection{Star Networks}
Star networks are common networks in the real world, where a single central agent is connected with many peripheral agents. Examples include a single boss and many subordinates, the head of a political party that coordinates the disparate branches of the party, or a large trader that deals with many small traders. There are several important forces to consider when placing agents within a star network. Such networks depend greatly on the central agent, because that agent is connected with all other agents and it therefore has the most links. The central agent is therefore the most important agent to consider, and choosing the best agent to be in the center is crucial to the overall welfare of the network. 

The initial mean and the signal precision of the central agent are two exogenous parameters that must be carefully considered when choosing the central agent. A high initial mean is beneficial because it increases the expected flow benefits that all the other agents who are connected to the central agent will receive. However, a higher signal precision is harmful because it allows for a greater probability that the central agent becomes ostracized quickly, thus causing the network to fall apart. Such an event would greatly lower social welfare. Therefore there is a trade off between the initial mean and the signal precision of the central agent: it is desirable to have a central agent with a higher mean but a lower signal precision. In particular, choosing the agent based only on its initial mean expected quality is not optimal, whereas under complete information it would be optimal to always place the highest realized quality agent in the center.

We show these results formally in the next proposition. For concreteness, suppose that the central agent in the network is denoted by agent $1$. The exogenous parameters of the agents are defined the same way as previously.

\begin{proposition}
\label{starprop}
The overall social welfare is strictly increasing in $\mu_1$ and strictly decreasing in $\tau_1$ and $\sigma^2_1$.
\end{proposition}
\begin{proof}
We can break social welfare into two components: the welfare of the central agent, and the welfares of each periphery agent. Notice that the welfares of the periphery agents are strictly increasing in $\mu_1$ but do not depend on $\sigma_1$ or $\tau_1$ for similar reasons as in the proof of Theorem \ref{compnl}. Also, the welfare of the central agent is strictly increasing in $\mu_1$ as that allows the central agent to stay in the network for a longer period of time. Thus overall social welfare is increasing in $\mu_1$. The welfare of the central agent is strictly decreasing in $\tau_1$ for the same reasons as in Proposition \ref{prop:welfsig}. Thus overall social welfare is decreasing in this parameter.
\end{proof}

Figure \ref{fig:starfig} shows the trade-off between the mean and the signal precision of the central agent explictly via a simulation. It plots the contour lines of the overall \textit{ex ante} welfare of the network, and it shows that social welfare increases as the initial mean increases and the signal precision decreases, and therefore selecting the best central agent depends on both factors.

We note that for the periphery agents on the other hand, the exogenous parameters have a much less clear relationship with the overall social welfare. We can actually show through examples that social welfare can increase or decrease in each of these factors for periphery agents. The same relationships as for the central agent can hold, and a simple example would be a two person network. However a marginally higher mean or a lower signal precision by a single periphery agent can actually \textit{decrease} overall welfare.  For instance, consider a network where the central agent has an initial expected quality close to $c$, one periphery agent denoted by agent $i$ also has an initial expected quality close to $c$, and the qualities of all other periphery agents is very high. In such a case, increasing the expected quality of agent $i$ by a small amount, or decreasing agent $i$'s signal precision would harm overall social welfare. These changes would result in the central agent being connected to agent $i$ for a longer stretch of time, which is undesirable since the other periphery agents are of much higher expected quality, and so causing the central agent to send more information is harmful. Therefore in such a network it would be better for agent $i$ to send information more quickly in order for it to exit the network sooner.

\begin{figure}
\centering
\includegraphics[width=.8\linewidth]{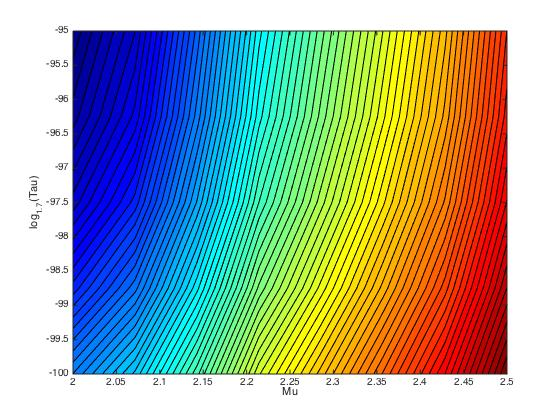}
\caption{Simulation Illustrating Proposition \ref{starprop}: The simulation uses a network of 6 agents. The 5 periphery agents have $\mu_i=2$, $\sigma_i^2=2$, and $\tau_i=1$. The central agent has $\sigma_1^2=2$, while its initial mean ranges from 2 to 2.5 and its signal precision ranges from 1.7 raised to the power of -100 to -95. 4000 realizations of agent hitting times were taken at each different agent mean, for a total of 32000 different hitting times. When drawing realizations across different means, the quantiles of the agent qualities were fixed to ensure faster convergence.}
\label{fig:starfig}
\end{figure}

We note that the trade-off identified above matters only if learning is fast enough, whereas if learning becomes very slow (or the designer becomes very impatient), then this trade-off goes away. This is summarized in the following proposition.

\begin{proposition}
\label{starlearnlimit}
If the rate of learning becomes very slow (i.e. $\lambda \to 0$), then the optimal star network is obtained by placing the agent with the highest initial reputation in the center.
\end{proposition}
\begin{proof}
In the limit of very slow learning, only the initial welfare generated matters, and placing the agent with the highest expected quality in the center generates the highest welfare over all initial network structures.
\end{proof}

Proposition \ref{starlearnlimit} shows that the decision to place an agent at the center depends only on each agent's initial mean in the limit of very slow learning learning (a similar result holds for very high designer impatience). When the network is constrained to be a star network, the highest initial welfare is obtained by having the highest initial expected quality agent in the center if the learning is very slow. 

\subsection{Ring networks}
In this section we focus on a special type of network: a ring network. Suppose for convenience that agents are homogeneous in terms of initial expected quality and variance. Assume that under the network constraint $\Omega$ each agent is limited to at most two neighbors. Hence, for a given number of agents, they would only be able to form one or multiple ring networks of different sizes. This could represent a work environment in which agents work in pairs on projects and can work on up to two projects at a time, or a financial network in which financial institutions seek two partners to trade with.

We study how the size of different rings affects the welfare an agent obtains and hence, we can determine the optimal size of the rings that agents should form together. Let $W(n)$ denote the welfare an agent can obtain if it is in a ring of size $n$ under the network constraint $\Omega$\footnote{For convenience we assume that the number of agents $N$ is divisible by $n$}. We show that networks with rings of three agents (so there is triadic closure among the agents) will maximize both agent welfare and overall social welfare.\footnote{The social networks literature views triadic closure as the result of common preferences or trust, whereas our model derives a reputational reason for such networks.}

\begin{proposition}
The optimal size of a ring network is 3 agents.
\end{proposition}
\begin{proof}
Consider a ring network consisting of three agents $i$, $j$, $k$. We focus on the welfare of agent $i$ and show that it is maximized compared to rings of other sizes. Since agents are identical, this means that total social welfare is maximized as well.

Agent $i$ obtains a positive benefit in two cases: (1) realizations in which both agents $j$ and $k$'s reputation never hit $c$; (2) realizations in which exactly one of agents $j$ and $k$'s reputation never hits $c$. The probabilities that these two cases happen are independent of the network structure by Proposition \ref{prop:dist}. In the first case, having additional agent(s) between agent $j$ and $k$ does not affect agent $i$'s realization and hence, agent $i$'s welfare is not affected. In the second case, having additional agent(s) between agent $j$ and $k$ will change $i$'s realization with positive probability. Consider a realization in which agent $k$'s reputation never hits $c$ and agent $j$'s reputation hits $c$ at $t_j$. In the ring of size 3, agent $j$'s direct neighbor besides $i$ (i.e. agent $k$) never hits $c$. When there are additional agents, it is either the case that agent $j$'s new direct neighbor never hits $c$ or hits $c$ before infinity. If agent $j$'s new direct neighbor hits $c$ before infinity, then agent $j$'s new hitting time may increase and hence agent $i$'s new hitting time may decrease, leading to a lower welfare for agent $i$.
\end{proof}

The intuition behind this result is similar to the reasoning of Proposition \ref{prop:linkwelf}, in which having a direct neighbor send more information is beneficial for an agent. With only three agents in each ring, an agent learns about a neighbor that would be excluded from the stable network at a faster rate, since that neighbor remains connected with the other neighbor, when the other neighbor is included in the stable network, until the first neighbor is ostracized. This guarantees a fast rate of learning about the low expected quality neighbor, allowing the agent itself to have more time to stay connected with the high expected quality neighbor that is not ostracized. With more than three agents, the neighbor that is excluded from the stable network may have its own neighbor disconnect in advance, slowing the rate of information the ostracized neighbor produces and hurting the agent itself. 

We can extend this result to ring networks with more than three agents. Similar to the odd/even effect highlighted in Corollary \ref{corr:uniquepath}, we can show that rings with an odd number of agents will always have higher expected social welfare than ring networks with an even number of agents. However, as the number of agents grows large the difference in the social welfare of an even and odd number of agents eventually goes to zero.

\begin{corollary}
If $n$ is odd, then $W(n) > W(m), \forall m > n$. If $n$ is even, then $W(n) < W(m), \forall m > n$. Moreover, $W(n)$ converges to a limit as $n$ approaches infinity.
\end{corollary}
\begin{proof}
The proof is similar to that of Proposition \ref{starprop} except we take into account the odd-even effect discussed in Corollary \ref{corr:uniquepath}. We still only need to consider the case when exactly one of agents $j$ and $k$'s reputation never hits $c$. Without loss of generality assume that $j$ is not included in the stable network. With four agents, social welfare is lower than with three because the neighbor of $j$, call it $l$, may be ostracized before than $j$ is ostracized, causing $j$'s information speed to slow down. With five agents, social welfare is higher than with four because in the same case, there is a chance that agent $l$'s other neighbor is ostracized before agent $l$ is ostracized, resulting in a decrease in agent $l$'s information speed and an increase in agent $j$'s information speed. This argument can be extended indefinitely for any number of agents to prove the above result. We note that the limits of the social welfares are the same, since the probability of a neighbor very far away sending a signal that affects agent $j$'s hitting time approaches zero as the number of agents becomes very large. Such an event can only occur if all the agents in between have an ostracism time less than agent $j$ itself, an event with probability that approaches zero as the number of agents gets large.
\end{proof}

\section{Extensions}
As seen above, learning can have a negative impact on social welfare in a variety of networks, and a large reason for this is the myopia of the agents. Since the agents are not experimenting for long enough, learning is inefficient and social welfare is lost. In this section, we consider four possible extensions that could alleviate this issue and allow for higher social welfare.

\subsection{Linking Subsidy}
A potential method of addressing the negative effects of learning is to give subsidies to the agents for linking with others. For instance, a company may wish to give workers awards or bonuses for collaborating with colleagues. Or in a financial setting, a regulator may give financial incentives for firms conducting mutual investments, or guarantee interbank transactions during a financial crisis to lower default risk. We model a subsidy by assuming that for every link that an agent maintains, it receives an extra flow benefit of $\delta$ from the network designer. This linking subsidy does not affect the social welfare computation since it is a direct transfer from the network designer to the agent, but it would change agents' decisions of when to break a link. Since agents are myopic, an agent $i$ will break its link with agent $j$ if and only if agent $j$'s reputation drops below $c - \delta$. The linking subsidy therefore causes the agents to learn more information about their neighbor's quality and break only if it is very likely to be bad. We show below that by properly choosing the linking subsidy the social welfare can improve compared with the case when there is no learning about agents' qualities. Let $W(\delta)$ denote the \textit{ex ante} social welfare when the linking subsidy is equal to $\delta$. 

\begin{theorem}
There exists $\bar{\delta}$ such that $\forall \delta > \bar{\delta}$, $W(\delta) > W^*$. Moreover, $\lim\limits_{\delta \to \infty} W(\delta) = W^*$.
\label{subsidy} 
\end{theorem} 
\begin{proof}
See appendix.
\end{proof}

Note that by Theorem \ref{compnl}, this result also shows that the social welfare is higher than the standard network model with no subsidy. Thus by imparting subsidies on agents to encourage them to experiment for longer, the social welfare is higher than previously. The intuition is that when the link subsidy is high enough, any link that is broken will involve an agent that is of really low expected quality. Thus although the agent that is ostracized may still hurt from being disconnected, its neighbors will benefit by a sufficiently large amount that overall social welfare increases. Therefore learning is now beneficial and improves welfare overall. The second part of the theorem states that if the linking subsidy becomes too high, then the social welfare will converge to the social welfare without learning. This is because when the subsidy is too high it becomes almost impossible for a link to break, and so the network with high probability will not change, just like in the case without learning. Therefore having a linking subsidy is beneficial for the network, but the subsidy cannot be set too high either in order to maximize social welfare.

\subsection{New Link Formation}
Another way that learning would be more socially beneficial is if agents were able to form new links with other agents whose reputations are very high. In this extension, we assume that a pair of agents who are not initially linked according to the network constraint $\Omega$ can form a new link by incurring an instantaneous cost $\gamma > 0$. There is no cost to forming links with agents that they are connected to under $\Omega$. So unlike previously when there was a hard barrier between agents not connected according to $\Omega$, agents can now break this barrier by paying an instantaneous cost. This cost could be exogenous, for instance the cost of time and energy in becoming familiar with a new agent, or the cost of reducing some physical barrier between the agents (distance or geographic barriers). The cost could also be set by the network designer such as a tax on link creation. Since we assume the formation cost is instantaneous, it is infinitesimal in the social welfare calculation and so only affects welfare through its impact on agent actions.

We assume that forming a link this way requires bilateral consent as usual. Agent $i$ will want to form a link with agent $j$ if agent $j$'s reputation is higher than $c + \gamma$. Therefore a new link between agents $i$ and $j$ is formed at time $t$ if and only if $\mu^t_i\geq c + \gamma$ and $\mu^t_j \geq c + \gamma$. The dynamics of our model will now feature some agents attaining high reputation levels and being able to link with other previously inaccessible agents that have also attained high reputation levels. Allowing these two high expected quality agents to link together will improve social welfare due to the large mutual benefits that are generated from their link. 

We can compare the social welfare produced by allowing this extra link formation against the social welfare in the basic model. Let $W(\gamma)$ denote the \textit{ex ante} social welfare when the link formation cost is equal to $\gamma$, and let $W$ be the social welfare in the basic model without the extra link formation.

\begin{theorem}
There exists $\bar{\gamma}$ such that $\forall \gamma \geq \bar{\gamma}$,  $W(\gamma) > W$.
\end{theorem}
\begin{proof}
Consider any realization $\varepsilon$ when link formation is not allowed. The \textit{ex post} welfare $W(\varepsilon)$ is changed only when there is some time $t^*$ such that there exist two agents $i$ and $j$, who are not initially connected, such that $\mu^{t^*}_i \geq c + \gamma$ and $\mu^{t^*}_j \geq c + \gamma$. In the original realization $\varepsilon$, conditional on $t^*$, there are two cases
\begin{itemize}
  \item $\zeta_1$: Both agents' reputations never hit $c$ after $t^*$. 
  \item $\zeta_2$: At least one agent's reputation hits $c$ after $t^*$.
\end{itemize}
When $\zeta_2$ occurs, allowing link formation may change the hitting time of all agents' in the network and hence, the welfare $W(\varepsilon|\zeta_2)$ may change. However, the probability of $\zeta_2$ occurring tends to zero as $\bar{\gamma}$ tends to infinity by Proposition \ref{prop:dist}. When $\zeta_1$ occurs, the social welfare increases by at least $\frac{e^{-\rho t^*}}{\rho}\frac{(c+\bar{\gamma}) - (1 - P(\zeta_1)) c}{P(\zeta_1)}$. When $\zeta_2$ occurs, the welfare decreases by at most $B(\zeta_2)$, a function that is at most linear in $\bar{\gamma}$ as it grows large, since the set of agents and their initial qualities are fixed. Thus the overall change in welfare can be written as

\begin{eqnarray}
W'(\varepsilon) - W(\varepsilon) \geq P(\zeta_1)\frac{e^{-\rho t^*}}{\rho}\frac{(c+\bar{\gamma}) - (1 - P(\zeta_1)) c}{P(\zeta_1)} - P(\zeta_2)e^{-\rho t^*} B(\zeta_2)
\end{eqnarray}

By choosing $\bar{\gamma}$ large enough, we can ensure that $P(\zeta_2)$ is small enough such that the change is positive in all such realizations $\varepsilon$. Therefore $W(\gamma) > W$. 

\end{proof}

This theorem states that if the link formation cost is high enough then the social welfare is improved over the base model because two agents that decide to form a new link will do so with high reputations. Thus the social welfare generated by a new link is likely to be high as well, and this dominates any potential informational externalities that the link could create. Note however that a $\gamma$ that is too low may actually harm welfare. For instance suppose there are a group of moderate expected quality agents that are all linked to a very high expected quality agent, but separated from each other according to $\Omega$. This is similar to the core-periphery setting examined in Theorem \ref{coreperiph}. In such a case, allowing moderate reputation agents to link with each other would cause them to harm each other via the negative informational effects of the link. This would reduce welfare overall compared to the base model. Therefore allowing for new link formation can improve welfare, but the threshold for the link being formed must be sufficiently high as well. The optimal $\bar{\gamma}$ would depend on the specific properties of the network. If as in the example there exists a group of very high reputation agents that the moderate reputation agents are linked with, then $\bar{\gamma}$ would likely be higher as well, as it becomes more important for moderate reputation agents to not be linked with each other.

\subsection{Agent Entry}
Our model can also be tractably extended to allow agents to enter into the network over time. Specifically, suppose that for the set of $N$ agents in $V$ there is a corresponding set of entry times $\{e_i\}_{i\in V}$, with $e_i\geq0\ \forall i$. Agents with $e_i=0$ are present in the network at the beginning, while agents who have $e_i>0$ enter later on. These entry times are fixed and known to the agents in the model. The network constraint $\Omega$ is now defined over the set of all $N$ potential agents and still specifies which agents are allowed to connect to each other, including agents that arrive later. This network constraint effectively determines where agents enter into the network at their entry times. The learning process is the same as before, with learning occurring for agents within the network based on their current amount of neighbors, and no learning occurring for an agent that has not yet entered. 

Agents still make decisions myopically and will connect with a neighbor for as long as that neighbor's reputation is above the connection cost. Since we assume all agents have initial reputations above the cost, an incumbent agent with always wish to connect with a newly entering agent. However, the new agent would not want to connect with one of its neighbors that has already been ostracized previously within the network. The dynamics will evolve similarly to before, with agents connecting to neighbors until a neighbor's reputation falls too low, at which point the neighbor will be ostracized.  The difference now is that new agents will arrive at certain times, and when they do they will change the benefits and amount of information produced by the network.

We can compare the model with agent entry against the base model where all agents were present in the beginning, i.e. $e_i=0 \ \forall i\in V$. We fix a network constraint $\Omega$ and perform comparative statics on the entry times of the agents. We first show that incorporating agent entry will not change either the set or the distribution of stable networks.

\begin{proposition}
The set of stable networks is unchanged with agent entry. The probability of each stable network emerging is the same as that given in Corollary \ref{corr:stableprob} and identical to the case without agent entry.
\end{proposition}
\begin{proof}
First note that Proposition \ref{prop:dist} still holds for each agent, regardless of the specific entry times. This is because the later entry of an agent only shifts the time at which it gets ostracized, but will not change the fact that it ever gets ostracized. Since the probability that each agent is ostracized is not affected, the set of stable networks and the probability that each stable network emerges does not change either. Thus the same probability distribution over stable networks as in corollary \ref{corr:limitprob} will result.
\end{proof}

Although the properties of the final stable networks are not affected by agent entry, the overall social welfare will be affected. It is possible to calculate social welfare in a similar method as in Theorem \ref{welfcalc}, as we can account for agent entry by rescaling the hitting times of the agents in the network appropriately. Incorporating agent entry has two separate effects on social welfare: first, the links that the entering agent has are started later, so the benefits from those links are realized later as well and thus discounted more heavily; second, the neighbors of the entering agent send less information before that agent enters, and the agent itself may send information more slowly if one of its neighbors is ostracized before it enters, delaying the time at which the agent and its neighbors are potentially ostracized from the network. The first effect hurts social welfare because the benefits from any link are positive in expectation. However, the second effect can improve social welfare by delaying the agents' ostracization times and increasing the benefits that each agent is able to extract from the network. It is possible for the second effect to dominate the first, so that delaying entry for an agent raises social welfare overall.

\begin{theorem}
\label{entrytheorem}
For some network parameters, increasing a single agent's entry time $e_i$ can increase social welfare.
\end{theorem}
\begin{proof}

\begin{figure}

\centering
\includegraphics[width=.8\linewidth]{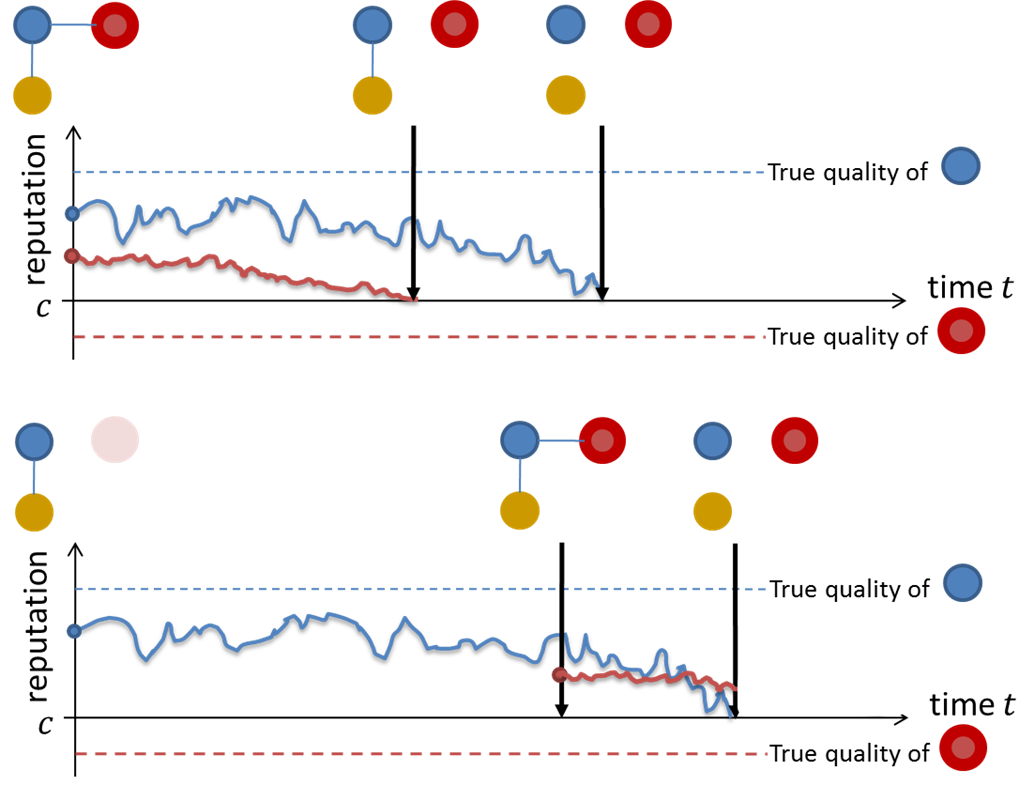}
\caption{Example for Theorem \ref{entrytheorem}}
\label{fig:entfig}
\end{figure}

We prove using an example, shown in Figure \ref{fig:entfig}. In this network of three agents, suppose that the white agent's expected quality is very high. Suppose both the red (large circle and bolded line) and the blue (small circle and thin line) agents expected qualities are very close to $c$. Since the white agent's expected quality is very high, the social welfare of the network will be completely determined by the amount of time the blue agent connects with the white agent. By delaying the entry of the red agent, the blue agent is able to stay connected for longer in each realization, and so social welfare increases.
\end{proof}

In the example of Figure \ref{fig:entfig}, note that although delaying the entry of the red (large circle and bolded line) agent is helpful, it is still better to have the agent enter at some finite time instead of never entering. This is because the blue agent's reputation will eventually converge to its true quality by the law of large numbers, and in the case where the blue (small circle and thin line) agent has a good true quality, enabling a link with the red agent will produce positive benefits. In addition, after waiting for a sufficiently large amount of time, the probability that the blue agent ever becomes ostracized if it hasn't already goes to zero, so the red agent is unlikely to impact the blue agent's connection with the white agent. Therefore delaying the entry of the red agent is beneficial, but the red agent should not be excluded from the network altogether.

\begin{figure}

\centering

\includegraphics[width=.8\linewidth]{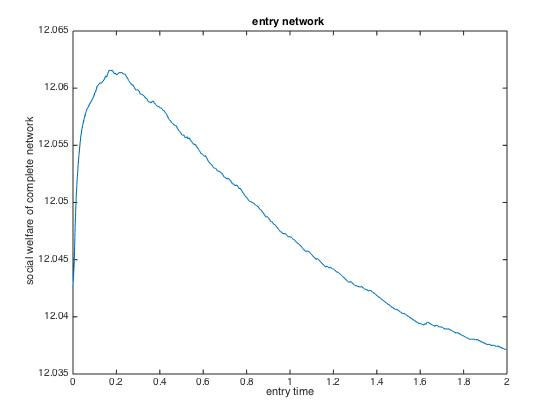}
\caption{Simulation Illustrating Theorem \ref{entrytheorem}: The simulation uses a network of 3 agents who are linked according to Figure \ref{fig:entfig}. All  agents have $\sigma_i^2=20$ and $\tau_i=1$, and the discount rate is 1. The high expected quality agent has an initial mean of 100, while the other two agents have initial means of 1. The entry time of the entering agent ranges from .0005 to 2 at increments of .05. The y-axis shows the average social welfare of the network at each entry time. 80000 draws were made of each agent's hitting time and true quality level, and the social welfare was then computed by varying the entry time.}
\label{fig:entsim}
\end{figure}

Figure \ref{fig:entsim} shows this trade-off explicitly via a simulation highlighting the example shown in Figure \ref{fig:entfig}. When the new agent enters later, social welfare initially increases because the incumbent agents have more time to benefit from their links. However, if the entry time becomes too large then the social welfare decreases, since the reputations of the incumbent agents have stabilized already, and it is thus better to have the new agent enter sooner and benefit from the network as well.

As an implication, a financial regulator may wish to delay new firms from entering the network in times of crisis when there is a lot of uncertainty, and then allow them to enter once the crisis has ended and reputations are more stable. Or in an organization, a firm may wish to not expand too quickly, and instead take the time to allow the current workers to better understand each other first.

\subsection{Agent Re-entry}
Our model can be tractably extended to allow agents to be forgiven and then let back into the network. For instance, suppose that a worker in a company can improve its quality after it becomes ostracized through some exogenous means, such as going back to school to increase its abilities, or taking counseling to better its personality. In financial networks, suppose that a bank can get recapitalized by the government after it gets shut out of the network, allowing its expected quality to increase. After the agent undergoes this exogenous process, the agent's reputation improves and so the other agents are again willing to link with it. We show that agent forgiveness in this manner can increase social welfare as well as mitigate the negative effects of learning. In fact, learning may now actually become beneficial.

We model agent forgiveness by assuming that when an agent is ostracized from the network, the agent can reenter the network at a later time. How long the agent must wait before reentry is an exogenous parameter, which we denote by $L$. When the agent reenters, its reputation is the same as the reputation that it started out with initially, $N(\mu_i,\sigma^2_i)$.\footnote{We make this assumption to avoid adding too many new exogenous parameters. Our results can be extended to a more general setting as well where the reputation changes upon re-entry.} As mentioned above, this re-entry could be the result of the agent undergoing additional training or preparation to improve its quality. An alternative interpretation is also possible where this is in fact a new agent entering the network, but from the same population or background as the original agent. Thus the new agent starts out with the same reputation as the original agent.

We assume that each agent can reenter into the network as long as it has not already been ostracized in the past a total of $R$ times. Therefore an agent can reenter the network as long as it has not already reentered $R-1$ times in the past. $R$ is an exogenous parameter that represents the number of times which ostracized agents are willing to undergo the process to improve themselves, or the number of replacement agents that can be brought into the network. A higher value of $R$ means that the ostracized agents are willing to undergo the improvement process even if they have been ostracized multiple times in the past.

With agent re-entry, we can still compute the set of stable networks, as well as the probability that each stable network emerges. The probability that an agent is included in the stable network is now equal to the probability that an agent does not get ostracized a total of $R$ times in a row. Since the agent's reputation is redrawn each time upon re-entry, this probability can be computed using the products of the probabilities in Proposition \ref{prop:dist}. The exact formula is given in the following proposition. Compared with the original probabilities, agent re-entry implies that each agent is more likely to be part of the stable network, since they have more chances with which to get a high true quality draw.

\begin{proposition} $P(S_i)$ depends only on the initial quality distribution and the link cost and can be computed by
\begin{equation}
P(S_i) = 1-\left(1-\int_{c}^\infty(1 - \exp(-\frac{2}{\sigma^2_i}(\mu_i - c)(q_i - c)))\phi\left((q_i - \mu_i)\frac{1}{\sigma_i}\right)dq_i \right)^R
\end{equation}`
\end{proposition}
\begin{proof}
The probability that an agent is ostracized permanently is found by taking the 1 minus the probability in Proposition \ref{prop:dist}, and then raising that to the power of R. Therefore the probability that an agent is included in the stable network is found by taking 1 minus this probability.
\end{proof}

Note that since this probability is very similar to the probability given in Proposition \ref{prop:dist}, all of the relationships between this probability and the exogenous parameters (initial mean, variance, signal precision, and link cost) highlighted in Corollary \ref{corr:limitprob} and Theorem \ref{sigprec} continue to hold. In addition, we can derive an analogue of Corollary \ref{corr:stableprob} using these new probabilities. Thus we can still characterize the explicit probability that any stable network emerges as time goes to infinity and the re-entry process by all the agents has concluded.

We can also derive results about agent welfare when re-entry is possible. Specifically, we can show that if the number of periods of re-entry $R$ is sufficiently large, and the time that an agent takes to reenter $L$ is sufficiently small, then learning becomes beneficial. This is intuitive, because if agents are learned about faster, then bad agents can exit the network sooner to undergo improvement while good agents will stay in and are unaffected. Having agent forgiveness mitigates the negative effects of learning, and makes learning a positive overall.

\begin{theorem}
If $R$ is sufficiently large compared to $\tau_i$, and $L$ is sufficiently small compared to $\tau_i$, then a small increase in $\tau_i$ increases the overall social welfare of the network.
\end{theorem}
\begin{proof}
Note that as $R$ converges to infinity, the probability that each agent is included in the stable network goes to $1$. Therefore the social welfare generated by any agent $i$ will depend on the first time instance at which it enters and does not become ostracized. This is because $L$ is very small, so agent $i$ loses very little benefit when it is ostracized. The first time at which agent $i$ reenters and does not get ostracized is strictly decreasing in its information precision $\tau_i$, since a faster information speed implies that it gets ostracized earlier later on.  Therefore a larger signal precision increases overall social welfare.
\end{proof}

We can extend the above result to show that a fully connected network is the optimal $\Omega$ when the network is very forgiving and the downtime of reentry is low. A fully connected network would allow all agents to link with each other and benefit from the resulting mutual interactions. In addition, since learning is now beneficial, the fact that each agent has many links in a fully connected network and thus sends a lot of information also increases social welfare. This result highlights the fact that with agent forgiveness, more densely connected networks can become optimal, and the designer can allow for more links in the initial networks.

\begin{theorem}
If $R$ is sufficiently large compared to $\tau_i$ for all $i$, and $L$ is sufficiently small compared to $\tau_i$ for all $i$, then a fully connected network is the optimal $\Omega$.
\end{theorem}
\begin{proof}
Similar to the above proof, note that as $R$ converges to infinity, the probability that each agent is included in the stable network goes to $1$, and so the social welfare generated by any agent $i$ will depend on the first time instance at which it enters and does not become ostracized. With a fully connected network, each agent has as many neighbors as possible and sends information very quickly, and so the timing of this first time instance becomes sooner. Notice also that each agent's flow payoff is positive at any point in time that they are in the network. Since $L$ is very small, agents are in the network almost continuously, and so having more links increases the flow benefits that each agent receives. Thus a fully connected network is the optimal initial network.
\end{proof}

\section{Conclusion}

This paper analyzed agent learning and the resulting network dynamics when there is incomplete information. We presented a highly tractable model that explicitly characterized what the set of stable networks are for a given network, showed how learning affects both individual and social welfare depending on the specific network topology, and analyzed what optimal initial network structures look like for different groups of agents. Our results shed new light on network dynamics in real world situations, and they offer guidelines for optimal network design when there is initial uncertainty about the agents. When agents are sufficiently myopic in their actions, ostracism becomes harmful not just for the ostracized agents themselves, but to all agents in an \textit{ex ante} fashion. A network designer should thus structure links appropriately in order to minimize the negative effects of ostracism.

Our results could be extended in several interesting ways. One natural extension would be to allow the qualities of agents to evolve over time. In the simplest extension, the agent's true quality $q_i$ itself change according to an exogenous stochastic process, for instance a Brownian motion. More interestingly, it would be natural to assume that the evolution of true quality depends endogenously on the information the agent receives so that agents who receive better information tend to develop higher true qualities and hence also generate better information in the future. Thus, the structure of the network and the true qualities of the agents in the network co-evolve. Higher reputation agents may link to agents that are also of higher reputation, and so their true qualities would improve as well, while lower reputation agents may struggle to find good agents to link with, and their true qualities would decline as a result. 

Other possible extensions include having private information among the agents instead of locally public information. In this way agents would learn about their neighbors at different rates, and so they may make different decisions when connecting or disconnecting with other agents. This result could mitigate the negative effects of learning, as information is different across link, and so having more links does not increase the rate of learning. Agent preferences could also be heterogeneous, which would further increase the diversity of links and the range of linking decisions. This is a topic we are currently researching in van der Schaar and Zhang (2015).

Finally, it would be interesting to allow agents to engage in games with their linked neighbors instead of merely generating flow benefits. Games played over networks have been analyzed in several papers within the networks literature (see Jackson and Zenou (2014) for a review), but never in a dynamic setting with learning such as that considered in the current paper. The game played by agents could be a prisoner's dilemma or another type of cooperation game where the payoffs depend on the agent's types. Agents would need to seek out other agents that they can achieve high payoffs in the game with, and this process would also require learning over time about a neighbor's type. As agents are able to learn each other's type more accurately, they may achieve greater efficiency in their plays and also sustain cooperation for a longer length of time.




\newpage
\appendix

\section*{Proof of Proposition \ref{prop:dist}}
\begin{proof}
Suppose for now that agent $i$'s reputation always evolves at the constant signal precision $\tau_i$. Then given the true quality $q_i$ for agent $i$, the probability that agent $i$'s reputation never hits $c$ before $t$ can be found using standard arguments (see for example Wang and P\"{o}tzelberger (1997)) and is given by
\begin{eqnarray}
P(S^t_i|q_i) = \Phi\left(\sqrt{t\tau_i}(q_i-c) + \frac{\frac{1}{\sigma^2_i}(\mu_i - c)}{\sqrt{t\tau_i}}\right)\\
-\exp(-\frac{2}{\sigma^2_i}(\mu_i - c)(q_i - c))\Phi\left(\sqrt{t\tau_i}(q_i - c) - \frac{\frac{1}{\sigma^2_i}(\mu_i-c)}{\sqrt{t\tau_i}}\right)
\end{eqnarray}
Therefore, given $q_i$, the probability that agent $i$ stays in the network is
\begin{equation}
P(S_i|q_i) = \lim\limits_{t\to\infty} P(S^t_i|q_i)
\end{equation}
\begin{itemize}
  \item If $q_i > c$, as $t \to \infty$, then we have $\Phi\left(\sqrt{t\tau_i}(q_i-c) + \frac{\frac{1}{\sigma^2_i}(\mu_i - c)}{\sqrt{t\tau_i}}\right) \to 1$ and\\ $\Phi\left(\sqrt{t\tau_i}(q_i - c) - \frac{\frac{1}{\sigma^2_i}(\mu_i-c)}{\sqrt{t\tau_i}}\right) \to 1$. Thus, $P(S_i|q_i) = 1 - \exp(-\frac{2}{\sigma^2_i}(\mu_i - c)(q_i - c))$, namely agent $i$ stays in the network with positive probability and the probability is increasing in the true quality $q_i$.
  \item If $q_i < c$, as $t \to \infty$, then we have $\Phi\left(\sqrt{t\tau_i}(q_i-c) + \frac{\frac{1}{\sigma^2_i}(\mu_i - c)}{\sqrt{t\tau_i}}\right) \to 0$ and\\ $\Phi\left(\sqrt{t\tau_i}(q_i - c) - \frac{\frac{1}{\sigma^2_i}(\mu_i-c)}{\sqrt{t\tau_i}}\right) \to 0$, thus $P(S_i|q_i)  = 0$, namely agent $i$'s reputation hits $c$ before $t = \infty$ for sure.
  \item If $q_i = c$, it is clear that $P(S_i^t|q_i) = 0$ as $t \to \infty$.
\end{itemize}

Taking the expectation over $q_i$, we have
\begin{equation}
P(S_i) = \int_{c}^\infty(1 - \exp(-\frac{2}{\sigma^2_i}(\mu_i - c)(q_i - c)))\phi\left((q_i - \mu_i)\frac{1}{\sigma_i}\right)dq_i/\sigma_i
\end{equation}
From the above expression we can see that $P(S_i)$ only depends on the initial quality distribution ($\mu_i$ and $\sigma_i$) and the link cost $c$ but does not depend on the Brownian motion precision $\tau_i$. Since breaking links only changes the Brownian motion precision, the probability that an agent's reputation never hits $c$ is independent of the initial network $G^0$ or the signal precision $\tau_i$.
\end{proof}

\section*{Proof of Corollary \ref{corr:limitprob}}
\begin{proof}
We first show that $P(S_i)$ is increasing in $\mu_i$. Let $q_i - \mu_i = x$. Then $P(S_i)$ can be rewritten as
\begin{equation}
P(S_i) = \int_{c-\mu_i}^\infty(1 - \exp(-\frac{2}{\sigma^2_i}(\mu_i - c)(\mu_i - c + x)))\phi\left(\frac{x}{\sigma_i}\right)dx/\sigma_i
\end{equation}
Consider a larger expected quality $\mu_i' > \mu_i$, we have
\begin{eqnarray}
P(S_i|\mu_i') = \int_{c-\mu_i'}^\infty(1 - \exp(-\frac{2}{\sigma^2_i}(\mu_i' - c)(\mu_i' - c + x)))\phi\left(\frac{x}{\sigma_i}\right)dx/\sigma_i\\
>\int_{c-\mu_i}^\infty(1 - \exp(-\frac{2}{\sigma^2_i}(\mu_i' - c)(\mu_i' - c + x)))\phi\left(\frac{x}{\sigma_i}\right)dx/\sigma_i\\
>\int_{c-\mu_i}^\infty(1 - \exp(-\frac{2}{\sigma^2_i}(\mu_i - c)(\mu_i - c + x)))\phi\left(\frac{x}{\sigma_i}\right)dx/\sigma_i = P(S_i|\mu_i)
\end{eqnarray}
Therefore, $P(S_i)$ is increasing in $\mu_i$.

Next we show that $P(S_i)$ is decreasing in $\sigma_i$.
\begin{eqnarray}
P(S_i) = \int_{c-\mu_i}^\infty\phi(\frac{x}{\sigma_i})dx/\sigma_i - \int_{c-\mu_i}^\infty e^{-\frac{2}{\sigma^2_i}(\mu_i - c)(\mu_i - c + x))}\frac{1}{\sigma_i\sqrt{2\pi}}e^{-\frac{1}{2}\left(\frac{x}{\sigma_i}\right)^2} dx\\
=\int_{c-\mu_i}^\infty\phi(\frac{x}{\sigma_i})dx/\sigma_i - \int_{c-\mu_i}^\infty \frac{1}{\sigma_i\sqrt{2\pi}}e^{-\frac{1}{2\sigma_i^2}(2(\mu_i -c) + x)^2)} dx\\
=\int_{c-\mu_i}^\infty\phi(\frac{x}{\sigma_i})dx/\sigma_i - \int_{\mu_i - c}^\infty\phi(\frac{x}{\sigma_i})dx/\sigma_i =  \int_{c-\mu_i}^{\mu_i-c}\phi(\frac{x}{\sigma_i})dx/\sigma_i
\end{eqnarray}
Therefore, $P(S_i)$ is decreasing in $\sigma_i$.

Finally, we show that $P(S_i)$ is decreasing in $c$. Consider a smaller $c' < c$, we have
\begin{eqnarray}
P(S_i|c) = \int_{c}^\infty(1 - \exp(-\frac{2}{\sigma^2_i}(\mu_i - c)(q_i - c)))\phi\left((q_i - \mu_i)\frac{1}{\sigma_i}\right)dq_i/\sigma_i\\
< \int_{c}^\infty(1 - \exp(-\frac{2}{\sigma^2_i}(\mu_i - c')(q_i - c')))\phi\left((q_i - \mu_i)\frac{1}{\sigma_i}\right)dq_i/\sigma_i\\
< \int_{c'}^\infty(1 - \exp(-\frac{2}{\sigma^2_i}(\mu_i - c')(q_i - c')))\phi\left((q_i - \mu_i)\frac{1}{\sigma_i}\right)dq_i/\sigma_i = P(S_i|c')
\end{eqnarray}
The first inequality is because for $q_i > c$, $1 - \exp(-\frac{2}{\sigma^2_i}(\mu_i - c)(q_i - c)) < 1 - \exp(-\frac{2}{\sigma^2_i}(\mu_i - c)(q_i - c))$. The second inequality is because for $c' < q_i < c$, $1 - \exp(-\frac{2}{\sigma^2_i}(\mu_i - c)(q_i - c)) > 0$.
\end{proof}

\newpage

\begin{figure}

\centering
\includegraphics[width=1\linewidth]{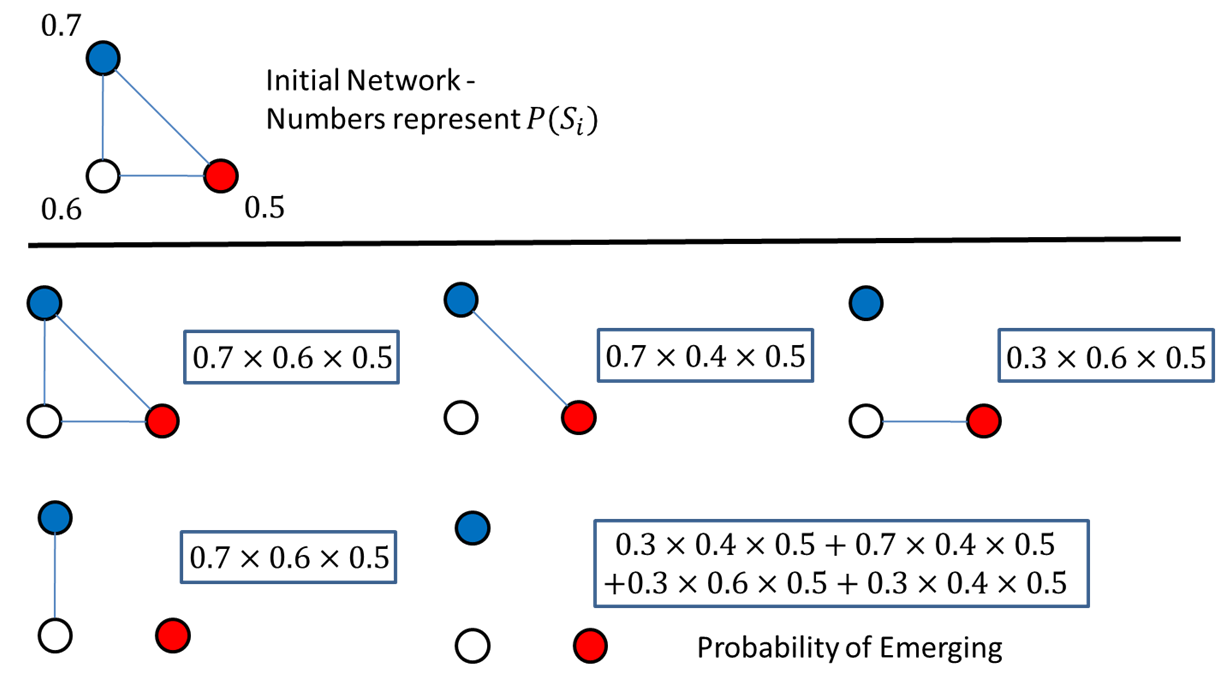}
\caption{Set of Stable Networks Given an Initial Network: This figure shows the five possible stable networks that could emerge given an initial network of three agents. In addition, $P(S_i)$ is given for all the agents, which allows us to calculate the exact probability of each of these networks emerging. For the first four networks, there is only one realization of $\{S_i, ^\neg S_i\}_{i\in V}$ that corresponds to it. For the last network, there are four possible realizations, one in which $^\neg S_i$ occurs for all agents, and three in which $S_i$ occurs for a single agent.}
\label{fig:compute_prob}
\end{figure}

\newpage

\section*{Proof of Lemma \ref{pdf}}
\begin{proof}
Since the Brownian motion precision is constant, using the survival probability
\begin{eqnarray}
P(S^t_i|q_i) = \Phi\left(\sqrt{t\tau_i}(q_i-c) + \frac{\frac{1}{\sigma^2_i}(\mu_i - c)}{\sqrt{t\tau_i}}\right)\\
-\exp(-\frac{2}{\sigma^2_i}(\mu_i - c)(q_i - c))\Phi\left(\sqrt{t\tau_i}(q_i - c) - \frac{\frac{1}{\sigma^2_i}(\mu_i-c)}{\sqrt{t\tau_i}}\right)
\end{eqnarray}
we can compute $f(\hat\varepsilon_i^{t}|q_i)=-\frac{d P(S^t_i|q_i)}{d t}$ as
\begin{eqnarray}
f(\hat\varepsilon_i^{t}|q_i)
=-\frac{1}{2}\left(\sqrt{\tau_i}(q_i-c)t^{-1/2}-\frac{\mu_i-c}{\sigma_i^2\sqrt{\tau_i}}t^{-3/2}\right) \phi\left(\sqrt{t\tau_i}(q_i-c) + \frac{\frac{1}{\sigma^2_i}(\mu_i - c)}{\sqrt{t\tau_i}}\right)\\
+e^{-\frac{2}{\sigma^2_i}(\mu_i - c)(q_i - c)}\frac{1}{2}\left(\sqrt{\tau_i}(q_i-c)t^{-1/2}+\frac{\mu_i-c}{\sigma_i^2\sqrt{\tau_i}}t^{-3/2}\right) \phi\left(\sqrt{t\tau_i}(q_i-c) - \frac{\frac{1}{\sigma^2_i}(\mu_i - c)}{\sqrt{t\tau_i}}\right)\\
=-\frac{1}{2}\left(\sqrt{\tau_i}(q_i-c)t^{-1/2}-\frac{\mu_i-c}{\sigma_i^2\sqrt{\tau_i}}t^{-3/2}\right) \phi\left(\sqrt{t\tau_i}(q_i-c) + \frac{\frac{1}{\sigma^2_i}(\mu_i - c)}{\sqrt{t\tau_i}}\right)\\
+\frac{1}{2}\left(\sqrt{\tau_i}(q_i-c)t^{-1/2}+\frac{\mu_i-c}{\sigma_i^2\sqrt{\tau_i}}t^{-3/2}\right) \phi\left(\sqrt{t\tau_i}(q_i-c) + \frac{\frac{1}{\sigma^2_i}(\mu_i - c)}{\sqrt{t\tau_i}}\right)\\
=\frac{\mu_i-c}{\sigma_i^2\sqrt{\tau_i}}t^{-3/2}\phi\left(\sqrt{t\tau_i}(q_i-c) + \frac{\frac{1}{\sigma^2_i}(\mu_i - c)}{\sqrt{t\tau_i}}\right)
\end{eqnarray}
Taking the expectation over $q_i$, we obtain $f(\hat\varepsilon_i^{t_i})$.
\end{proof}

\newpage

\section*{Algorithm for Computing Hitting Time Mapping Function $M$}

\begin{figure}
\centering
\includegraphics[width=1\linewidth]{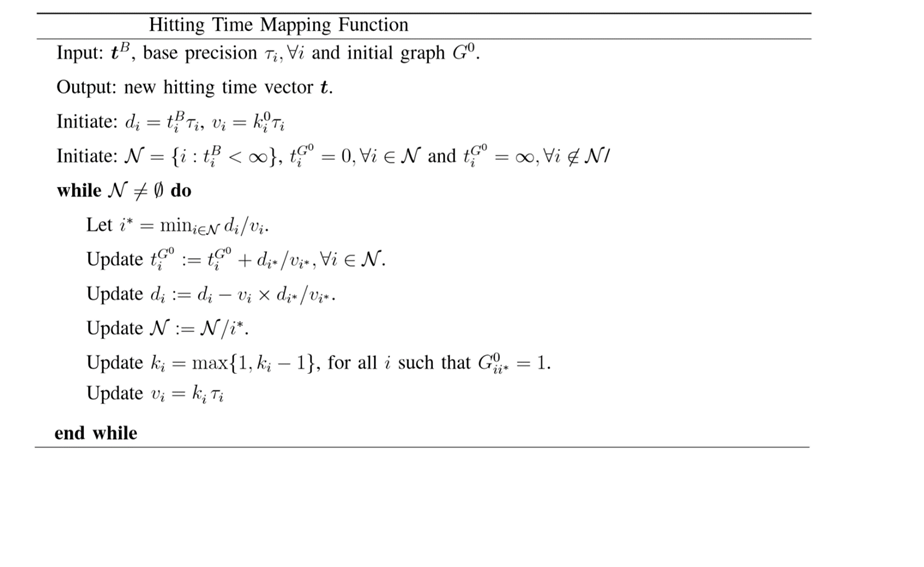}
\caption{}
\label{fig:algorithm1}
\end{figure}

\newpage

\section*{Proof of Theorem \ref{compnl}}
\begin{proof}
Consider the \textit{ex ante} surplus $W_{ij}$ that agent $i$ obtains from the link with a neighbor $j$. The \textit{ex ante} welfare for agent $i$ is simply the summation of this surplus over all $j$ that $i$ is linked with. $W_{ij}$ can be computed as
\begin{equation}
W_{ij} = \int_{q}\int_{0}^\infty e^{-\rho t} P(L_{ij}^t|q)(q_j - c) dt \phi(q) dq
\end{equation}
where $P(L_{ij}^t|q)$ is the probability that the link between $i$ and $j$ still exists at time $t$. Let $t^*$ be the time at which the link between $i$ and $j$ is broken. Then the social welfare can be computed as
\begin{eqnarray}
W_{ij} =\int_{q}\int_{0}^\infty e^{-\rho t}(q_j - c) dt \phi(q) dq - E_{t^*}[\int_{t^*}^\infty e^{-\rho t}E_{q_j} (q_j - c|t\geq t^*) dt]\\
=\int_{0}^\infty e^{-\rho t}(\mu_j - c) dt - E_{t^*}[\int_{t^*}^\infty e^{-\rho t}E_{q_j} (q_j - c|t\geq t^*) dt]
\end{eqnarray}
where the expectation is taken over the realizations in which the hitting time is $t^*$.
The second term can be further decomposed. Let $t^*_i$ denote the case when $t^* = t_i$, namely agent $i$'s reputation hits $c$ before agent $j$, and $t^*_j$ be the case where $t^* = t_j$, namely agent $j$'s reputation hits $c$ before agent $i$. Then
\begin{equation}
W_{ij} = \int_{0}^\infty e^{-\rho t}(\mu_j - c) dt \\
- E_{t^*_i}[\int_{t^*_i}^\infty e^{-\rho t}E_{q_j} (q_j - c|t\geq t^*_i) dt] - E_{t^*_j}[\int_{t^*_j}^\infty e^{-\rho t}E_{q_j} (q_j - c|t\geq t^*_j) dt]
\end{equation}
In the case of $t^*_j$, for any $t\geq t^*_j$, since the learning has stopped, $E_{q_j}(q_c - c|t \geq t^*_j) = 0$ by the definition of $t^*_j$. Similarly, $E_{q_j}(q_j - c|t \geq t^*_i) > 0$ because at $t^*_i$ the expected quality of $q_j$ is strictly greater than $c$ since agent $j$ has not been ostracized. Therefore,
\begin{equation}
W_{ij} = W^*_{ij} - E_{t^*_i}[\int_{t^*_i}^\infty e^{-\rho t}E_{q_j} (q_j - c|t\geq t^*_i) dt] < W^*_{ij}
\end{equation}
Summing over all $j$ that $i$ is linked with, we conclude that the agent $i$'s \textit{ex ante} welfare with learning is strictly less than that when there is no learning. Note that this result holds independently of the values of ${\tau_1,...,\tau_N}$, as the signal precisions affect only the distribution of agent hitting times, but not the expected quality of the agents conditional on ever being ostracized.
\end{proof}

\section*{Proof of Proposition \ref{limitlearn}}
\begin{proof}
From Theorem \ref{welfcalc}, we have that the \textit{ex ante} social welfare is given by:
\begin{equation}
W = E_{\hat\varepsilon} \sum\limits_{i}\left(\frac{1-e^{-\rho M_i(t)}}{\rho} \sum\limits_{j:g^{0}_{ij} = 1, t_j = \infty}\frac{\mu_j-c}{P(S_j)}\right)
\end{equation}
If the designer is completely impatient, it only cares about the social surplus at time $0$ since with probability approaching $1$ no links will be broken among the agents. Since all agents' expected qualities are above the linking cost, having all agents connected with each other yields the highest social surplus. Similarly if learning becomes very slow, then the agent's reputations are never updated and the same reasoning applies. In both cases the $e^{-\rho M_i(t)}$ term approaches zero in the above equation regardless of the network structure, and so adding more agents increases welfare.

If the designer is completely patient, only the stable networks matter. Since the stable network does not depend on the speed of learning and the probability that an agent stays in the stable network is independent of others by Proposition \ref{prop:dist}, having all agents connected with each other leads to the maximum number of links in the stable networks and hence the highest social surplus. Similarly if learning becomes very fast, the stable network will always be reached immediately and the same reasoning applies. In both cases the $e^{-\rho M_i(t)}$ term approaches one in the above equation regardless of the network structure, and so adding more agents increases welfare.
\end{proof}

\section*{Proof of Theorem \ref{subsidy}}
\begin{proof}
(1) Consider the welfare on link $l_{ij}$.  As in the proof of Theorem \ref{compnl}, let $t_i^*$ denote the event in which agent $i$'s reputation hits $c - \delta$ at time $t_i^*$ before agent $j$. The \textit{ex ante} welfare of link $l_{ij}$ can be computed as
\begin{eqnarray}
W_{ij} +W_{ji} = \int_0^\infty e^{-\rho t}(\mu^0_i + \mu^0_j - 2c)dt\\
-E_{t_i^*}[ \int_{t^*}^\infty e^{-\rho t} E_{q_i, q_j}(q_i + q_j - 2c| t \geq t_i^*)dt]\\
-E_{t_j^*}[ \int_{t^*}^\infty e^{-\rho t} E_{q_i, q_j}(q_i + q_j - 2c| t \geq t_j^*)dt]
\end{eqnarray}
Note that the first integral in the above equation represents $W_{ij}^*+W_{ji}^*$, the social welfare of the link without learning. 

In the case of $t_i^*$, for any $t\geq t_i^*$, since the learning has stopped, $ E_{q_i, q_j}(q_i - c|t \geq t_i^*) = c - \delta - c = -\delta$ by the definition of $t^*$. Since agent $j$ is not ostracized, we would have $E_{q_i, q_j}(q_j - c|t \geq t_i^*) > -\delta$. Let $h(\delta, t_i^*) = E_{q_i, q_j}(q_i + q_j - 2c| t \geq t_i^*) = E_{q_i, q_j}(q_j |t \geq t_i^*) - 2c -\delta$. This is the net change in flow payoff after the link is severed. We will show that for any $t_i^*$, $h(\delta, t_i^*) < 0$ if $\delta$ is sufficiently large. A symmetric argument then establishes that $h(\delta, t_j^*) < 0$, and the two together imply that the welfare of the link with learning is greater than $W_{ij}^*+W_{ji}^*$. Then, adding up over all links shows that the overall social welfare is higher than that without learning.

To prove that $h(\delta, t_i^*) < 0$ if $\delta$ is sufficiently large, we will show that $E_{q_i, q_j}(q_j|t \geq t_i^*)$ is bounded above for any $t_i^*$ as $\delta$ tends to infinity. Consider any \textit{ex post} realization of $t_i^*$, which implies that agent $j$'s reputation does not hit $c - \delta$ before $t_i^*$. There are two possibilities for agent $j$'s reputation (here we assume that agent $j$ continues sending information at its fixed signal precision if all its other neighbors are ostracized, as in section 4):
\begin{itemize}
  \item $\zeta_1$: it never hits $c - \delta$ after $t_i^*$ either.
  \item $\zeta_2$: it hits $c - \delta$ at some time after $t_i^*$.
\end{itemize}
Clearly, $E(q_j|\zeta_1) > E(q_j|\zeta_2) = c - \delta$. Hence $E_{q_i, q_j}(q_j| t \geq t^*) < E(q_j|\zeta_1)$. The value of $E(q_j|\zeta_1)$ is given by equation (6) in the text, with $c$ replaced by $c-\delta$. When $\delta \to \infty$, using equation (6) we can show that $\lim\limits_{\delta \to \infty} E(q_j|\zeta_1) = \mu^0_j$ through the application of L'Hopital's rule. 

Therefore, $\forall \epsilon > 0$, there exists $\delta'_{ij}$ such that $\forall \delta > \delta'_{ij}$, $E(q_j|\zeta_1) - \mu^0_j < \epsilon$. Hence, fix a value of $\epsilon > 0$ and let $\bar{\delta}_{ij} = \max\{\delta'_{ij}, \mu^0_j - 2c + \epsilon\}$, which ensures for all $\delta > \bar{\delta}_{ij}$, $E(q_j|\zeta_1) - 2c - \delta < 0$. This also implies that $h(\delta, t_i^*) < 0$ for all $t_i^*$ and $\delta>\bar{\delta}_{ij}$. By choosing $\bar{\delta} = \max_{i,j}\bar{\delta}_{ij}$, we ensure the overall \textit{ex ante} social welfare is greater than $W^*$.

(2) Define $H_{ij}(\delta) = E_{t_i^*}[ \int_{t_i^*}^\infty e^{-\rho t} E_{q_i, q_j}(q_i + q_j - 2c| t \geq t_i^*)dt]$. We will prove $\lim\limits_{\delta \to \infty} H_{ij}(\delta) = 0$. To prove this, we will show that for any sequence $\delta_n \to \infty$, the sequence $H_{ij}(\delta_n) \to 0$. We divide $H_{ij}(\delta)$ into two parts,
\begin{eqnarray}
H_{ij}(\delta) = E_{t^*_i < \hat{t}(\delta)}\left[ \int_{t^*_i}^\infty E_{q_i, q_j} [e^{-\rho t}(q_i + q_j - 2c|t\geq t^*_i)] dt\right] \\
+ E_{t^*_i \geq \hat{t}(\delta)}\left[ \int_{t^*_i}^\infty E_{q_i, q_j} [e^{-\rho t}(q_i + q_j - 2c|t\geq t^*_i)] dt\right] \\= H_{ij}'(\delta) + H_{ij}''(\delta)
\end{eqnarray}
for some $\hat{t}(\delta)$. We will find a sequence $\hat{t}(\delta_n)$ such that both $H_{ij}'(\delta_n) \to 0$ and $H_{ij}''(\delta_n) \to 0$ as $\delta_n \to \infty$.

Let $\hat{t}(\delta_n) = \delta_n$. First we will show that for $\delta_n$ large enough, $P(t_i^* < \delta_n) < \frac{1}{\delta_n^2}$. Note for a given $q_i$, the probability that the agent is ostracized before time $\delta_n$ is equal to:

\begin{eqnarray}
1- P(S^{\delta_n}_i|q_i) = 1 - \Phi\left(\sqrt{\delta_n\tau_i}(q_i-c+\delta_n) + \frac{\frac{1}{\sigma^2_i}(\mu_i - c+\delta_n)}{\sqrt{\delta_n\tau_i}}\right)\\
-\exp(-\frac{2}{\sigma^2_i}(\mu_i - c+\delta_n)(q_i - c+\delta_n))\Phi\left(\sqrt{\delta_n\tau_i}(q_i - c+\delta_n) - \frac{\frac{1}{\sigma^2_i}(\mu_i-c+\delta_n)}{\sqrt{\delta_n\tau_i}}\right)
\end{eqnarray}

Note that $\lim_{x\rightarrow\infty} \Phi(x)=1-\frac{e^{\frac{-x^2}{2}}}{x\sqrt{2\pi}}$. Therefore the term above approaches zero faster than $\frac{1}{\delta_n^2}$ as $\delta_n\rightarrow\infty$. Integrating over all $q_i$ shows that  $P(t_i^* < \delta_n) < \frac{1}{\delta_n^2}$ for large $\delta_n$.

Now consider $H'(\delta_n)$, it is bounded by
\begin{eqnarray}
|H'(\delta_n)| < P(t^*_i < \delta_n) \sup_{t^*_i < \delta_n}|\int_{t^*_i}^\infty E_{q_i, q_j} [e^{-\rho t}(q_i + q_j - 2c|t\geq t^*_i)] dt| \\
< \frac{\sup_{t^*_i < \delta_n}|\int_{t^*_i}^\infty E_{q_i, q_j} [e^{-\rho t}(q_i + q_j - 2c|t\geq t^*_i)] dt| }{\delta^2_n}\\
<\frac{1}{\rho \delta^2_n}\sup_{t^*_i < \delta_n}|E[q_j|\zeta_1] - c + \delta_n|
\end{eqnarray}
Since as $\delta_n \to \infty, E[q_j|\zeta_1] \to \mu^0_j$, we conclude that $|H'(\delta_n)| \to 0$.

Consider $H''(\delta_n)$, it is bounded by
\begin{eqnarray}
|H''(\delta_n)| <  \sup_{t^*_i \geq \delta_n}|\int_{t^*_i}^\infty E_{q_i, q_j} [e^{-\rho t}(q_i + q_j - 2c|t\geq t^*_i)] dt|\\
< \frac{1}{e^{\rho \delta_n}} \sup_{t^*_i > \delta_n}|\int_{t^*_i}^\infty E_{q_i, q_j} [e^{-\rho (t-\delta_n)}(q_i + q_j - 2c|t\geq t^*_i)]|\\
<\frac{1}{\rho e^{\rho \delta_n}} \sup_{t^*_i < \delta_n}|E[q_j|\zeta_1] - c + \delta_n|t \geq t^*_i|
\end{eqnarray}
Similarly, since as $\delta_n \to \infty, E[q_j|t\geq t^*_i] \to \mu^0_j$, we conclude that $|H''(\delta_n)| \to 0$.
\end{proof}

\newpage
\bibliographystyle{IEEEtran}
\bibliography{networkbib}

\end{document}

%% file: macros.tex
 \newcommand{\cbranch}[9]    {
  \begin{picture}(\pw,\pht)(-\xi,-\yi)
               \put(0,200)    {#1}                 
   \ifx#2S     \put(40,85)    {\line(0,1)  {100}}  
   \else\ifx#2D\multiput(27,85)(26,0){2}           
                              {\line(0,1)  {100}} \fi \fi
   \ifx#4Q     \put(-305,0)   {\makebox(300,87)[r]{#3}}
   \else       \put(-455,0)   {\makebox(300,87)[r]{#3}} \fi
   \ifx#4S     \put(-150,33)  {\line(1,0)  {140}}  
   \else\ifx#4D\multiput(-150,20)(0,26){2}         
                              {\line(1,0)  {140}} \fi \fi
               \put(0,0)      {#5}                 
   \ifx#6S     \put(90,33)    {\line(1,0)  {140}}  
   \else\ifx#6D\multiput(90,20)(0,26){2}           
                              {\line(1,0)  {140}} \fi \fi
               \put(240,0)    {#7}                 
   \ifx#8S     \put(40,-15)   {\line(0,-1) {100}}  
   \else\ifx#8D\multiput(27,-15)(26,0){2}          
                              {\line(0,-1) {100}} \fi \fi
               \put(0,-210)   {#9}                 
  \end{picture}           }   
 \newcommand{\tbranch}[7]    {
  \begin{minipage}{\xbox pt}
   \begin{tabbing}
    $#3$\= $#4$\+ \kill
           $#1$   \\ [-#7pt]     
    \ifx#2S                      
           \hspace{4pt}\rule{0.4pt}{8pt}  \\ [-#7pt]\fi
    \ifx#2D
           \hspace{2pt}\rule{0.4pt}{8pt}  
           \hspace{-2pt}\rule{0.4pt}{8pt} \\ [-#7pt]\fi
           \- \kill
    $#3$\> $#4$\+ \\ [-#7pt]     
    \ifx#5S                      
           \hspace{4pt}\rule{0.4pt}{8pt}  \\ [-#7pt]\fi
    \ifx#5D
           \hspace{2pt}\rule{0.4pt}{8pt}    
           \hspace{-2pt}\rule{0.4pt}{8pt} \\ [-#7pt]\fi
           $#6$
  \end{tabbing}
 \end{minipage}    }             
 \newcommand{\cright}[7]        {
   \begin{picture}(\pw,\pht)(-\xi,-\yi)
    \ifx#2Q \put(-305,-15){\makebox(300,87)[r]{#1}}   
     \else  \put(-455,-15){\makebox(300,87)[r]{#1}} \fi
    \ifx#2S \put(-150,33)    {\line(1,0)   {140}} \fi 
    \ifx#2D \put(-150,20)    {\line(1,0)   {140}}     
            \put(-150,46)    {\line(1,0)   {140}} \fi 
    \put(0,0)                {#3}                     
    \ifx#4S \put(80,70)      {\line(1,1)   {100}} \fi 
    \ifx#4D \put(71,79)      {\line(1,1)   {100}}
            \put(89,61)      {\line(1,1)   {100}} \fi 
    \put(185,170)            {#5}                     
    \ifx#6S \put(80,0)       {\line(1,-1)  {100}} \fi 
    \ifx#6D \put(71,-9)      {\line(1,-1)  {100}}     
            \put(89,9)       {\line(1,-1)  {100}} \fi 
    \put(185,-170)           {#7}                     
   \end{picture}             }               
 \newcommand{\cleft}[7]      {
   \begin{picture}(\pw,\pht)(-\xi,-\yi)
    \put(-405,160)           {\makebox(300,87)[r]{#1}} 
    \ifx#2S \put(0,70)       {\line(-1,1)  {100}} \fi  
    \ifx#2D \put(9,79)       {\line(-1,1)  {100}}      
            \put(-9,61)      {\line(-1,1)  {100}} \fi  
    \put(0,0)                {#3}                      
    \ifx#4S \put(0,0)        {\line(-1,-1) {100}} \fi  
    \ifx#4D \put(-9,9)       {\line(-1,-1) {100}}      
            \put(9,-9)       {\line(-1,-1) {100}} \fi  
    \put(-405,-185)          {\makebox(300,87)[r]{#5}} 
    \ifx#6S \put(90,33)      {\line(1,0)   {140}} \fi  
    \ifx#6D \put(90,20)      {\line(1,0)   {140}}      
            \put(90,46)      {\line(1,0)   {140}} \fi  
    \put(240,0)              {#7}                      
   \end{picture}             }               
 \newcommand{\chemup}[7]        {
   \begin{picture}(\pw,\pht)(-\xi,-\yi)
    \put(-430,130)           {\makebox(300,87)[r]{#1}} 
    \ifx#2S \put(0,70)       {\line(-5,3)  {121}} \fi  
    \ifx#2D \put(7,81)       {\line(-5,3)  {121}}      
            \put(-7,59)      {\line(-5,3)  {121}} \fi  
    \put(0,0)                {#3}                      
    \ifx#4S \put(33,-10)     {\line(0,-1)  {140}} \fi  
    \ifx#4D \put(20,-10)     {\line(0,-1)  {140}}      
            \put(46,-10)     {\line(0,-1)  {140}} \fi  
    \put(0,-230)             {#5}                      
    \ifx#6S \put(80,70)      {\line(5,3)   {121}} \fi  
    \ifx#6D \put(73,81)      {\line(5,3)   {121}}      
            \put(87,59)      {\line(5,3)   {121}} \fi  
    \put(210,140)            {#7}                      
  \end{picture}              }                
 \newcommand{\cdown}[7]      {
   \begin{picture}(\pw,\pht)(-\xi,-\yi)
    \put(0,230)              {#1}                      
    \ifx#2S \put(33,80)      {\line(0,1)   {140}} \fi  
    \ifx#2D \put(20,80)      {\line(0,1)   {140}}      
            \put(46,80)      {\line(0,1)   {140}} \fi
    \put(0,0)                {#3}                      
    \ifx#4S \put(0,0)        {\line(-5,-3) {121}} \fi  
    \ifx#4D \put(-7,11)      {\line(-5,-3) {121}}      
            \put(7,-11)      {\line(-5,-3) {121}} \fi  
    \put(-430,-150)          {\makebox(300,87)[r]{#5}} 
    \ifx#6S \put(80,0)       {\line(5,-3)  {121}} \fi  
    \ifx#6D \put(87,11)      {\line(5,-3)  {121}}      
            \put(73,-11)     {\line(5,-3)  {121}} \fi  
    \put(210,-140)           {#7}                      
   \end{picture}             }               

 \newcommand{\threering}[9]     {
     \begin{picture}(\pw,\pht)(-\xi,-\yi)
       \put(300,0)     {\line(-3,-5) {150}}     
       \put(150,-244)  {\line(-3,5)  {150}}     
       \put(0,0)       {\line(1,0)   {300}}     
       \ifx#7D \put(40,-40){\line(1,0){220}}\fi 
       \ifx#1Q                                  
         \else\put(300,0)    {\line(5,3)   {128}}
              \put(433,50)   {#1}           \fi
       \ifx#2Q                                  
         \else\ifx#9C \put(150,-244) {\line(0,-1){100}}
                      \put(114,-424) {#2}  
                 \else\put(150,-244) {\line(5,-3)  {128}}
                      \put(283,-344) {#2} \fi \fi 
       \ifx#2Q                                  
         \else\ifx#5Q \put(150,-244) {\line(0,-1){100}}
                      \put(114,-424) {#2}  
                 \else\put(150,-244) {\line(5,-3)  {128}}
                      \put(283,-344) {#2} \fi \fi 
       \ifx#3Q                                  
         \else\put(0,0)      {\line(-5,3)  {128}}
              \put(-430,34) {\makebox(300,87)[r]{#3}} \fi
       \ifx#4Q                                  
         \else\put(300,0)    {\line(5,-3)  {128}} 
              \put(433,-100) {#4}           \fi
       \ifx#5Q                                  
         \else\put(150,-244) {\line(-5,-3) {128}} 
              \put(-280,-360){\makebox(300,87)[r]{#5}}  \fi
       \ifx#6Q                                  
         \else\put(0,0)      {\line(-5,-3) {128}} 
              \put(-430,-116){\makebox(300,87)[r]{#6}} \fi
       \ifx#8Q
         \else\multiput(135,-244)(30,0){2}      
              {\line(0,-1)   {100}}             
              \put(114,-424) {#8}           \fi 
       \ifx#9C \put(150,-90){\circle{120}}   
               \put(110,-120){+}            \fi
       \end{picture}            }   
 \newcommand{\fourring}[9]      {
     \begin{picture}(\pw,\pht)(-\xi,-\yi)
      \put(300,300)   {\line(0,-1)  {300}}      
      \put(300,0)     {\line(-1,0)  {300}}      
      \put(0,0)       {\line(0,1)   {300}}      
      \put(0,300)     {\line(1,0)   {300}}      
      \ifx#7D\put(260,260){\line(0,-1){220}}\fi 
      \ifx#8D\put(40,40)  {\line(0,1) {220}}\fi 
      \ifx#1Q                                   
        \else\put(300,300){\line(5,3) {128}}  
             \put(433,350){#1}              \fi
      \ifx#2Q                                   
        \else\put(300,0)  {\line(5,-3){128}}
             \put(433,-100){#2}             \fi
      \ifx#3Q                                   
        \else\put(0,0)    {\line(-5,-3){128}}
             \put(-430,-116){\makebox(300,87)[r]{#3}} \fi
      \ifx#4Q                                   
        \else\put(0,300)  {\line(-5,3) {128}}
             \put(-430,334){\makebox(300,87)[r]{#4}}  \fi
      \ifx#5Q                                   
        \else\put(300,300){\line(5,-3) {128}}   
             \put(433,200){#5}              \fi
      \ifx#6Q                                   
        \else\put(0,300)  {\line(-5,-3){128}}   
             \put(-430,184){\makebox(300,87)[r]{#6}}  \fi
      \ifx#9Q                                   
        \else\multiput(280,-5)(20,30){2}        
             {\line(1,-1){100}} \put(405,-140){#9}   \fi 
    \end{picture}         }             
   \newcommand{\fivering}[9]    {
       \begin{picture}(\pw,\pht)(-\xi,-\yi)
        \put(342,200)    {\line(0,-1)   {200}}       
        \put(342,0)      {\line(-5,-3)  {171}}       
        \put(171,-103)   {\line(-5,3)   {171}}       
        \put(0,0)        {\line(0,1)    {200}}       
        \put(0,200)      {\line(1,0)    {342}}       
        \ifx#1Q                                      
          \else\put(342,200)   {\line(5,3)  {128}}
               \put(475,250)   {#1}              \fi
        \ifx#2Q                                      
          \else\put(342,0)     {\line(5,-3) {128}}
               \put(475,-100)  {#2}              \fi
        \ifx#3Q                                      
          \else\put(171,-103)  {\line(0,-1) {100}}
               \put(150,-283)  {#3}              \fi
        \ifx#4Q                                      
          \else\put(0,0)       {\line(-5,-3){128}}
               \put(-430,-116) {\makebox(300,87)[r]{#4}}  \fi
        \ifx#5Q                                      
          \else\put(0,200)     {\line(-5,3) {128}}
               \put(-430,234)  {\makebox(300,87)[r]{#5}}  \fi
        \ifx#6D\put(316,174)   {\line(0,-1) {148}}   
          \else\ifx#6S
                 \else\put(342,200) {\line(5,-3) {128}}
                      \put(475,100){#6}   \fi        
        \fi                                          
        \ifx#7D\put(26,26)     {\line(0,1)  {148}}   
          \else\ifx#7S
                 \else\put(0,200){\line(-5,-3){128}} 
                      \put(-430,84){\makebox(300,87)[r]{#7}}  \fi
        \fi                                          
        \ifx#8Q                                      
          \else\multiput(156,-103)(30,0){2}          
               {\line(0,-1) {100}} \put(135,-283){#8}  \fi  
        \ifx#9C\put(171,60)    {\circle{210}}        
               \put(130,35)    {$-$}       \fi       
      \end{picture}          }                 
  \newcommand{\sixring}[9]   {
   \begin{picture}(\pw,\pht)(-\xi,-\yi)
    \put(342,200)  {\line(0,-1)  {200}}          
    \put(342,0)    {\line(-5,-3) {171}}          
    \put(171,-103) {\line(-5,3)  {171}}          
    \put(0,0)      {\line(0,1)   {200}}          
    \put(0,200)    {\line(5,3)   {171}}          
    \put(171,303)  {\line(5,-3)  {171}}          
    \ifx#7D                                      
      \put(316,174)   {\line(0,-1)  {148}}       
      \else \ifx#7S                              
            \else\put(342,200) {\line(5,-3) {128}}
                 \put(475,100) {#7}    \fi  \fi
    \ifx#8D
      \put(162,-67)   {\line(-5,3)  {126}}       
      \else\ifx#8S
           \else\put(156,-203) {\line(0,1){100}} 
                \put(186,-203) {\line(0,1){100}} 
                \put(135,-283) {#8}    \fi  \fi
    \ifx#9D
      \put(36,191)    {\line(5,3)   {126}} \fi   
    \ifx#9C \put(171,100) {\circle{250}}   \fi   
    \ifx#1Q                                      
      \else\put(342,200)  {\line(5,3){128}}
           \put(475,250)  {#1}             \fi
    \ifx#2Q
      \else\put(342,0)    {\line(5,-3){128}}     
           \put(475,-100)    {#2}          \fi
    \ifx#3Q
      \else\put(171,-203) {\line(0,1){100}}
           \put(150,-283) {#3}             \fi   
    \ifx#4Q
       \else\put(0,0)     {\line(-5,-3){128}}    
            \put(-430,-116){\makebox(300,87)[r]{#4}}  \fi
    \ifx#5Q
      \else\put(0,200)    {\line(-5,3){128}}     
           \put(-430,234) {\makebox(300,87)[r]{#5}}   \fi
    \ifx#6Q
      \else\put(171,303)  {\line(0,1){100}}      
           \put(150,410)  {#6}             \fi
    \ifx#7D  \ifx#9C \message{Error: ring double bond 
                             simultaneous with circle}
    \fi   \fi
   \end{picture}           }                
  \newcommand{\sixringa}[9]   {
     \begin{picture}(\pw,\pht)(-\xi,-\yi)
       \put(342,200)  {\line(0,-1)  {200}}          
       \put(342,0)    {\line(-5,-3) {171}}          
       \put(171,-103) {\line(-5,3)  {171}}          
       \put(0,0)      {\line(0,1)   {200}}          
       \put(0,200)    {\line(5,3)   {171}}          
       \put(171,303)  {\line(5,-3)  {171}}          
       \ifx#7D                                      
         \put(306,191)   {\line(-5,3)  {126}}       
         \else \ifx#7S                              
               \else\put(342,200) {\line(5,-3) {128}}
                    \put(475,100) {#7}    \fi  \fi  
       \ifx#8D
         \put(178,-67)   {\line(5,3)   {126}}       
         \else\ifx#8S
              \else\put(156,-203) {\line(0,1){100}} 
                   \put(186,-203) {\line(0,1){100}} 
                   \put(135,-283) {#8}    \fi   \fi 
       \ifx#9D
         \put(26,26)     {\line(0,1)   {148}} \fi   
       \ifx#9C \put(171,100) {\circle{250}}   \fi   
       \ifx#1Q                                      
         \else\put(342,200)  {\line(5,3){128}}
              \put(475,250)  {#1}             \fi
       \ifx#2Q
         \else\put(342,0)    {\line(5,-3){128}}     
              \put(475,-100)    {#2}          \fi
      \ifx#3Q
         \else\put(171,-203) {\line(0,1){100}}
              \put(150,-283) {#3}             \fi   
       \ifx#4Q
          \else\put(0,0)     {\line(-5,-3){128}}    
               \put(-430,-116){\makebox(300,87)[r]{#4}} \fi
       \ifx#5Q
         \else\put(0,200)    {\line(-5,3){128}}     
              \put(-430,234) {\makebox(300,87)[r]{#5}}  \fi
       \ifx#6Q
         \else\put(171,303)  {\line(0,1){100}}      
              \put(150,410)  {#6}             \fi
       \ifx#7D  \ifx#9C \message{Error: ring double bond with 
                                        circle} \fi   \fi
      \end{picture}           }             
  \newcommand{\sixringb}[9]   {
  \newcommand{\di}   {\put(316,174){\line(0,-1){148}}} 
  \newcommand{\dii}  {\put(178,-67){\line(5,3) {126}}} 
  \newcommand{\diii} {\put(162,-67){\line(-5,3){126}}} 
  \newcommand{\dfour}{\put(26,26)  {\line(0,1) {148}}} 
  \newcommand{\dv}   {\put(36,191) {\line(5,3) {126}}} 
  \newcommand{\dsix} {\put(306,191){\line(-5,3){126}}} 
     \begin{picture}(\pw,\pht)(-\xi,-\yi)
       \put(342,200)  {\line(0,-1)  {200}}          
       \put(342,0)    {\line(-5,-3) {171}}          
       \put(171,-103) {\line(-5,3)  {171}}          
       \put(0,0)      {\line(0,1)   {200}}          
       \put(0,200)    {\line(5,3)   {171}}          
       \put(171,303)  {\line(5,-3)  {171}}          
       \ifx#7Q                                      
        \else\put(158,303) {\line(0,1) {100}}       
             \put(184,303) {\line(0,1) {100}}       
             \put(130,420) {#7}    \fi              
       \ifx#8Q
        \else\put(156,-203){\line(0,1){100}}        
             \put(186,-203){\line(0,1){100}}        
             \put(135,-283) {#8}   \fi              
       \ifcase#9 \put(171,100) {\circle{250}}        
         \or \di \or \dii \or \or diii \or \di \diii 
         \or \or \or dfour \or \dfour \di            
         \or \dfour \dii \or \or \or \or \or \or \dv 
         \or \dv \di \or \dv \dii \or \or \dv \diii  
         \or \dv \diii \di \or \or \or \or \or \or   
         \or \or \or \or \or \dsix \or \or \dsix \dii 
         \or \or \dsix \diii \or \or \or             
         \or \dsix \dfour \or \or \dsix \dfour \dii  
       \fi
       \ifx#1Q                                      
         \else\put(342,200)  {\line(5,3){128}}
              \put(475,250)  {#1}       \fi
       \ifx#2Q
         \else\put(342,0)    {\line(5,-3){128}}     
              \put(475,-100)    {#2}    \fi
      \ifx#3Q
         \else\put(171,-203) {\line(0,1){100}}
              \put(150,-283) {#3}       \fi         
       \ifx#4Q
          \else\put(0,0)     {\line(-5,-3){128}}    
               \put(-430,-116){\makebox(300,87)[r]{#4}}  \fi
       \ifx#5Q
         \else\put(0,200)    {\line(-5,3){128}}     
              \put(-430,234) {\makebox(300,87)[r]{#5}}   \fi
       \ifx#6Q
         \else\put(171,303)  {\line(0,1){100}}      
              \put(150,410)  {#6}       \fi
       \ifx#7D  \ifx#9C \message{Error: ring double bond with
                                        circle}  \fi  \fi
      \end{picture}           }              
  \newcommand{\chair}[8]       {
     \begin{picture}(\pw,\pht)(-\xi,-\yi)
      \put(0,0)       {\line(3,-4)  {170}}      
      \put(170,-226)  {\line(3,1)   {403}}      
      \put(573,-91)   {\line(6,-1)  {210}}      
      \put(783,-126)  {\line(-3,4)  {170}}      
      \put(613,100)   {\line(-3,-1) {403}}      
      \put(210,-35)   {\line(-6,1)  {210}}      
      \put(0,0)       {\line(0,1)   {100}}      
      \put(170,-226)  {\line(0,-1)  {100}}      
      \put(573,-91)   {\line(0,1)   {100}}      
      \put(783,-126)  {\line(0,-1)  {100}}      
      \put(613,100)   {\line(0,1)   {100}}      
      \put(210,-35)   {\line(0,-1)  {100}}      
      \put(0,0)       {\line(-3,-1) {128}}      
      \put(170,-226)  {\line(-3,1)  {128}}      
      \put(573,-91)   {\line(3,-1)  {128}}      
      \put(783,-126)  {\line(3,1)   {128}}      
      \put(613,100)   {\line(3,-1)  {128}}      
      \put(210,-35)   {\line(-3,1)  {128}}      
      \ifx#1Q\else\put(-30,110)     {#1}    \fi 
      \ifx#2Q\else\put(140,-406)    {#2}    \fi 
      \ifx#3Q\else\put(543,9)       {#3}    \fi 
      \ifx#4Q\else\put(753,-306)    {#4}    \fi 
      \ifx#5Q\else\put(-430,-85) {\makebox(300,87)[r]{#5}}
      \fi                                       
      \ifx#6Q\else\put(-260,-226){\makebox(300,87)[r]{#6}}
      \fi                                       
      \ifx#7Q\else\put(415,-230) {\makebox(300,87)[r]{#7}}
      \fi                                       
      \ifx#8Q\else\put(916,-115)    {#8}    \fi 
     \end{picture}              }         
 \newcommand{\naphth}[9]        {       
    \begin{picture}(\pw,\pht)(-\xi,-\yi)
     \multiput(0,200)(342,0){2}   {\line(5,3) {171}} 
     \multiput(171,303)(342,0){2} {\line(5,-3){171}} 
     \multiput(0,0)(342,0){3}     {\line(0,1) {200}} 
     \multiput(0,0)(342,0){2}     {\line(5,-3){171}} 
     \multiput(171,-103)(342,0){2}{\line(5,3) {171}} 
     \ifx#1Q
       \else\put(513,303)         {\line(0,1) {100}}  
            \put(492,410)         {#1}    \fi
     \ifx#2Q
       \else\put(684,200)         {\line(5,3) {128}}  
            \put(817,250)         {#2}    \fi
     \ifx#3Q
       \else\put(684,0)           {\line(5,-3){128}}  
            \put(817,-100)        {#3}    \fi
     \ifx#4Q
       \else\put(513,-103)        {\line(0,-1){100}}  
            \put(492,-283)        {#4}    \fi
     \ifx#5Q
       \else\put(171,-103)        {\line(0,-1){100}}  
            \put(150,-283)        {#5}    \fi
     \ifx#6Q
       \else\put(0,0)             {\line(-5,-3){128}} 
            \put(-430,-116) {\makebox(300,87)[r]{#6}}   \fi  
     \ifx#7Q
       \else\put(0,200)           {\line(-5,3) {128}} 
            \put(-430,234)  {\makebox(300,87)[r]{#7}}   \fi
     \ifx#8Q
       \else\put(171,303)         {\line(0,1)  {100}} 
            \put(150,410)         {#8}    \fi
     \ifx#9S                                          
       \else\put(316,174)         {\line(0,-1) {148}} 
            \put(162,-67)         {\line(-5,3) {126}} 
            \put(36,191)          {\line(5,3)  {126}} 
     \fi
     \ifx#9D\put(648,191)         {\line(-5,3) {126}} 
            \put(520,-67)         {\line(5,3)  {126}} 
     \fi
   \end{picture}       }                
 \newcommand{\terpene}[9]    {
 \begin{picture}(\pw,\pht)(-\xi,-\yi)
  \put(0,0)       {\line(5,1)  {196}}        
  \put(196,39)    {\line(5,-2) {186}}        
  \put(382,-35)   {\line(2,5)  {66}}         
  \put(448,130)   {\line(-5,2) {186}}        
  \put(262,204)   {\line(-5,-1){196}}        
  \put(66,165)    {\line(-2,-5){66}}         
  \put(196,39)    {\line(0,1)  {330}}        
  \put(262,204)   {\line(-2,5) {66}}         
  \ifx#1Q                                    
    \else\put(262,204)   {\line(4,3)  {120}}
         \put(387,267)   {#1}            \fi
  \ifx#2Q                                    
    \else\put(448,130)   {\line(4,3)  {120}}
         \put(573,193)   {#2}            \fi
  \ifx#3Q                                    
    \else\put(382,-35)   {\line(5,-2) {120}}
         \put(507,-121)  {#3}            \fi
  \ifx#4Q                                    
    \else\put(196,39)    {\line(2,-5) {56}}
         \put(231,-180)  {#4}            \fi
  \ifx#5Q                                    
    \else\put(0,0)       {\line(-5,-2){139}}
         \put(-441,-95)  {\makebox(300,87)[r]{#5}}  \fi
  \ifx#6Q                                    
    \else\put(66,165)    {\line(-4,3) {120}}
         \put(-362,216)  {\makebox(300,87)[r]{#6}}  \fi
  \ifx#7M\put(196,369)   {\line(4,3)  {120}} 
         \put(196,369)   {\line(-4,3) {120}} 
         \put(321,432)   {\small ${\rm CH_3}$}           
         \put(-226,425)                      
         {\makebox(300,87)[r]{\small ${\rm CH_3}$}} \fi
  \ifx#8Q\else
  \ifx#8O\put(443,142)   {\line(5,2)  {130}} 
         \put(453,118)   {\line(5,2)  {130}} 
         \put(583,155)   {$O$}               
    \else\put(448,130)   {\line(5,-2) {120}} 
         \put(573,44)    {#8}  \fi       \fi 
  \ifx#9Q\else
  \ifx#9D\put(368,-1)    {\line(2,5) {46}}   
    \else\put(382,-35)   {\line(4,3) {120}}  
         \put(507,28)    {#9}    \fi     \fi 
 \end{picture}    }   
 \newcommand{\steroid}[9]        {
     \begin{picture}(\pw,\pht)(-\xi,-\yi)
       \multiput(0,0)(342,0){3}    {\line(0,1) {200}} 
       \multiput(513,303)(342,0){3}{\line(0,1) {200}} 
       \multiput(0,200)(342,0){3}  {\line(5,3) {171}} 
       \multiput(0,0)(342,0){2}    {\line(5,-3){171}} 
       \multiput(171,-103)(342,0){2}{\line(5,3) {171}} 
       \multiput(171,303)(342,0){2}{\line(5,-3){171}} 
       \multiput(513,503)(342,0){2}{\line(5,3) {171}} 
       \multiput(684,606)(342,0){2}{\line(5,-3){171}} 
       \put(855,303)               {\line(1,0) {342}} 
       \put(855,503)               {\line(0,1) {128}} 
       \put(795,638)               {\small ${\rm CH_3}$}
       \ifx#1D\put(36,191)         {\line(5,3)  {126}}
        \else\ifx#1Q                                  
              \else\put(520,514)   {\line(-5,3)  {128}}
                   \put(506,492)   {\line(-5,3)  {128}}
                   \put(83,547)    {\makebox(300,87)[r]{#1}} \fi
       \fi                            
       \ifx#2D\put(162,-67)        {\line(-5,3) {126}} 
        \else\ifx#2Q                                   
              \else\put(-7,11)     {\line(-5,-3) {128}}
                   \put(7,-11)     {\line(-5,-3) {128}}
                   \put(-430,-116) {\makebox(300,87)[r]{#2}} \fi
       \fi                            
       \ifx#3Q
        \else\put(0,0)       {\line(-5,-3){128}}           
             \put(-430,-116) {\makebox(300,87)[r]{#3}} \fi 
       \ifx#4D\put(178,-67)  {\line(5,3)  {126}}     \fi 
       \ifx#5D\put(378,9)    {\line(5,-3) {126}}         
        \else\ifx#5Q
               \else\multiput(1011,606)(30,0){2} {\line(0,1) {100}}
                    \put(985,713)         {#5}         \fi
       \fi                             
       \ifx#6D\put(316,174)  {\line(0,-1) {148}}     \fi 
       \ifx#6M\put(342,200)  {\line(0,1)  {128}}         
              \put(282,335)  {\small ${\rm CH_3}$}     \fi
       \ifx#7Q
        \else\put(1026,606)  {\line(0,1)  {100}}         
             \put(995,713)   {#7}                    \fi 
       \ifx#8Q
        \else\put(1026,791)  {\line(0,1)  {100}}         
             \put(995,900)   {#8}                    \fi 
       \ifx#9Q
        \else\put(1026,606)  {\line(1,0)  {128}}         
             \put(1160,575)  {#9}                    \fi 
     \end{picture}       }                    

 \newcommand{\hetthree}[8]       {
  \begin{picture}(\pw,\pht)(-\xi,-\yi)
   \put(170,-170)  {#8}                        
   \put(0,0)       {C}                         
   \put(360,0)     {C}                         
   \put(80,30)     {\line(1,0)  {270}}         
   \put(70,-10)    {\line(1,-1) {100}}         
   \put(350,-10)   {\line(-1,-1){100}}         
   \ifx#1Q                                     
     \else\put(210,-180) {\line(0,-1){80}}     
          \put(180,-340) {#1}               \fi
   \ifx#6S\put(-10,30)   {\line(-1,0){140}}    
          \put(-460,-10) {\makebox(300,87)[r]{#2}} 
     \else\ifx#6H\put(-90,30)   {\line(-1,0) {140}}
                 \put(-70,0)    {H}                
                 \put(-540,-10) {\makebox(300,87)[r]{#2}}
          \else  \put(-310,-10) {\makebox(300,87)[r]{#2}} \fi \fi
   \ifx#7S\put(440,30)   {\line(1,0) {140}}    
          \put(590,0)    {#3}                  
     \else\ifx#7H\put(430,0)    {H}            
                 \put(510,30)   {\line(1,0)  {140}}
                 \put(660,0)    {#3}
          \else  \put(445,0)    {#3}   \fi \fi 
   \ifx#4Q
     \else\put(40,80)    {\line(0,1) {140}}    
          \put(-225,220) {\makebox(300,87)[r]{#4}} \fi 
   \ifx#5Q
     \else\put(400,80)   {\line(0,1) {140}}    
          \put(360,230)  {#5}              \fi 
  \end{picture}   }      
 \newcommand{\hetifive}[9]       {
   \begin{picture}(\pw,\pht)(-\xi,-\yi)
    \put(342,0)          {\line(-5,-3) {140}}         
    \put(342,0)          {\line(0,1)   {200}}         
    \put(342,200)        {\line(-1,0)  {342}}         
    \put(0,200)          {\line(0,-1)  {200}}         
    \put(0,0)            {\line(5,-3)  {140}}         
    \ifx#1Q
     \else\put(171,-137) {\line(0,-1)  {83}}          
          \put(135,-283) {#1}                     \fi 
    \ifx#2Q
     \else\ifx#2O\put(349,11) {\line(5,-3) {128}}     
                 \put(335,-11){\line(5,-3) {128}}     
                 \put(475,-120)            {O}        
          \else\put(342,0)    {\line(5,-3) {128}}     
               \put(475,-100) {#2}                \fi 
    \fi                                               
    \ifx#3Q
     \else\put(342,200)  {\line(5,3)   {128}}         
          \put(475,250)  {#3}                     \fi 
    \ifx#4Q
     \else\put(0,200)    {\line(-5,3)  {128}}         
          \put(-430,234) {\makebox(300,87)[r]{#4}}\fi 
    \ifx#5Q
     \else\ifx#5O\put(-7,11)  {\line(-5,-3) {128}}    
                 \put(7,-11)  {\line(-5,-3) {128}}    
                 \put(-200,-130)            {O}       
          \else\put(0,0)      {\line(-5,-3) {128}}    
               \put(-430,-116){\makebox(300,87)[r]{#5}} \fi 
    \fi                                               
    \ifx#6D\put(316,26)       {\line(0,1)  {148}} \fi 
    \ifx#7D\put(316,174)      {\line(-1,0) {290}} \fi 
    \ifx#8D\put(26,174)       {\line(0,-1) {148}} \fi 
    \put(135,-130)            {#9}                    
  \end{picture}               }     
 \newcommand{\heticifive}[9]     {
   \begin{picture}(\pw,\pht)(-\xi,-\yi)
    \put(342,0)          {\line(-5,-3) {140}}         
    \put(342,0)          {\line(0,1)   {160}}         
    \put(0,200)          {\line(1,0)   {300}}         
    \put(0,200)          {\line(0,-1)  {200}}         
    \put(0,0)            {\line(5,-3)  {140}}         
    \ifx#1Q
     \else\put(171,-137) {\line(0,-1)  {83}}          
          \put(135,-283) {#1}                     \fi 
    \ifx#2Q
     \else\ifx#2O\put(349,11)   {\line(5,-3) {128}}   
                 \put(335,-11)  {\line(5,-3) {128}}   
                 \put(475,-120) {O}                   
          \else\put(342,0)      {\line(5,-3) {128}}   
               \put(475,-100)   {#2}              \fi 
    \fi
    \ifx#3Q
     \else\put(370,217)  {\line(5,3)    {100}}        
          \put(475,250)  {#3}                     \fi 
    \ifx#4Q
     \else\ifx#4O\put(-7,189)   {\line(-5,3) {128}}   
                 \put(7,211)    {\line(-5,3) {128}}   
                 \put(-200,250) {O}                   
           \else\put(0,200)     {\line(-5,3) {128}}   
                \put(-430,234)  {\makebox(300,87)[r]{#4}} \fi
    \fi                                               
    \ifx#5Q
     \else\put(0,0)      {\line(-5,-3)  {128}}        
          \put(-430,-116){\makebox(300,87)[r]{#5}}\fi 
    \ifx#6D\put(316,26)  {\line(0,1)    {130}}    \fi 
    \ifx#7Q
     \else\ifx#7D\put(26,174)   {\line(0,-1) {148}}   
          \else\put(0,0)        {\line(-5,3) {128}}   
               \put(-430,34)    {\makebox(300,87)[r]{#7}}  \fi
    \fi                                               
    \put(135,-130)              {#8}                  
    \put(310,170)               {#9}                  
   \end{picture}                }    
 \newcommand{\pyrazole}[8]      {
   \begin{picture}(\pw,\pht)(-\xi,-\yi)
    \put(200,-84)        {\line(5,3)    {110}}        
    \put(342,200)        {\line(0,-1)   {140}}        
    \put(342,200)        {\line(-1,0)   {342}}        
    \put(0,200)          {\line(0,-1)   {200}}        
    \put(0,0)            {\line(5,-3)   {140}}        
    \put(135,-130)       {N}                          
    \put(310,-30)        {N}                          
    \ifx#1Q
     \else\put(171,-137) {\line(0,-1)   {83}}         
          \put(150,-283) {#1}                     \fi 
    \ifx#2Q
     \else\put(370,-17)  {\line(5,-3)   {100}}        
          \put(475,-100) {#2}                     \fi 
    \ifx#3Q
     \else\ifx#3O\put(335,211)  {\line(5,3)  {128}}   
                 \put(349,189)  {\line(5,3)  {128}}   
                 \put(475,250)  {O}                   
          \else\put(342,200)    {\line(5,3)  {128}}   
               \put(475,250)    {#3}              \fi 
    \fi
    \ifx#4Q
     \else\put(0,200)    {\line(-5,3)   {128}}        
          \put(-430,234) {\makebox(300,87)[r]{#4}}\fi 
    \ifx#5Q
     \else\ifx#5O\put(-7,11)    {\line(-5,-3){128}}   
                 \put(7,-11)    {\line(-5,-3){128}}   
                 \put(-200,-130){O}                   
          \else\put(0,0)        {\line(-5,-3){128}}   
               \put(-430,-116)  {\makebox(300,87)[r]{#5}}  \fi
    \fi                                               
    \ifx#6D\put(316,174) {\line(0,-1)   {114}}    \fi 
    \ifx#7D\put(316,174) {\line(-1,0)   {290}}    \fi 
    \ifx#8D\put(26,174)  {\line(0,-1)   {148}}    \fi 
  \end{picture}          }         
 \newcommand{\hetisix}[9]       {              
    \begin{picture}(\pw,\pht)(-\xi,-\yi)
      \put(342,0)     {\line(-5,-3)  {140}}    
      \put(0,0)       {\line(5,-3)   {140}}    
      \put(342,0)     {\line(0,1)    {200}}    
      \put(342,200)   {\line(-5,3)   {171}}    
      \put(171,303)   {\line(-5,-3)  {171}}    
      \put(0,200)     {\line(0,-1)   {200}}    
      \ifx#7D                                   
        \put(316,26)  {\line(0,1)    {148}}    
        \else\ifx#7Q
             \else\put(349,11)  {\line(5,-3){128}} 
                  \put(335,-11) {\line(5,-3){128}} 
                  \put(475,-120){#7}         \fi   
      \fi
      \ifx#8D
        \put(36,191)  {\line(5,3)    {126}} \fi  
      \ifx#1D
        \put(36,9)     {\line(5,-3)  {110}}      
        \else\ifx#1Q
             \else\put(171,-220) {\line(0,1) {83}}
                  \put(150,-283) {#1}       \fi  
      \fi
      \put(135,-130)  {#9}                       
      \ifx#2Q                                    
        \else\put(342,0)    {\line(5,-3)  {128}} 
             \put(475,-100)   {#2}
      \fi      
      \ifx#3Q
        \else\put(342,200)  {\line(5,3)   {128}} 
             \put(475,250)  {#3}
      \fi
      \ifx#4Q
        \else\put(171,303)  {\line(0,1)   {100}} 
             \put(150,410)  {#4}
      \fi
      \ifx#5Q
        \else\put(0,200)    {\line(-5,3)  {128}} 
             \put(-430,234) {\makebox(300,87)[r]{#5}}
      \fi
      \ifx#6Q
        \else\put(0,0)      {\line(-5,-3) {128}} 
             \put(-430,-116){\makebox(300,87)[r]{#6}}
      \fi                             
    \end{picture}     }               
  \newcommand{\pyrimidine}[9]        {          
     \begin{picture}(\pw,\pht)(-\xi,-\yi)
       \put(342,0)     {\line(-5,-3)  {140}}   
       \put(0,0)       {\line(5,-3)   {140}}   
       \put(342,0)     {\line(0,1)    {160}}   
       \put(171,303)   {\line(5,-3)   {140}}   
       \put(171,303)   {\line(-5,-3)  {171}}   
       \put(0,200)     {\line(0,-1)   {200}}   
       \put(135,-130)  {N}
       \put(310,170)   {N}
       \ifx#1Q
        \else\put(171,-137)  {\line(0,-1)  {83}} 
             \put(150,-283)  {#1}               \fi
       \ifx#2Q
        \else\ifx#2O\put(349,11)  {\line(5,-3) {128}}  
                    \put(335,-11) {\line(5,-3) {128}}  
                    \put(475,-120){O}                  
             \else\put(342,0)     {\line(5,-3) {128}}  
                  \put(475,-100)  {#2}           \fi   
       \fi
       \ifx#3Q
        \else\put(370,217)   {\line(5,3)   {100}}      
             \put(475,250)   {#3}                \fi   
       \ifx#4Q
        \else\ifx#4O\put(158,303) {\line(0,1)  {100}}  
                    \put(184,303) {\line(0,1)  {100}}  
                    \put(130,410) {O}                  
             \else\put(171,303)   {\line(0,1)  {100}}  
                  \put(150,410)    {#4}           \fi  
       \fi                                             
       \ifx#5Q
        \else\put(0,200)     {\line(-5,3)   {128}}     
             \put(-430,234)  {\makebox(300,87)[r]{#5}} \fi 
       \ifx#6Q
        \else\ifx#6O\put(-7,11)   {\line(-5,-3){128}}  
                    \put(7,-11)   {\line(-5,-3){128}}  
                    \put(-210,-130){O}                 
             \else\put(0,0)       {\line(-5,-3){128}}  
                  \put(-430,-116) {\makebox(300,87)[r]{#6}} \fi 
       \fi                                             
       \ifx#7D\put(306,9)    {\line(-5,-3)  {120}} \fi 
       \ifx#8D\put(180,267)  {\line(5,-3)   {120}}     
        \else\ifx#8Q
             \else\put(0,200){\line(-5,-3)  {128}}     
                  \put(-430,84) {\makebox(300,87)[r]{#8}}   \fi 
       \fi                                             
       \ifx#9D\put(26,174)   {\line(0,-1)   {148}} \fi 
   \end{picture}  }                     
 \newcommand{\pyranose}[9]   {
   \begin{picture}(\pw,\pht)(-\xi,-\yi)
    \put(588,0)         {\line(-1,-1) {159}}          
    \put(429,-159)      {\line(-1,0)  {270}}          
    \put(159,-159)      {\line(-1,1)  {159}}          
    \put(0,0)           {\line(1,1)   {159}}          
    \put(159,159)       {\line(1,0)   {225}}          
    \put(394,130)       {\small O}                    
    \put(460,130)       {\line(1,-1)  {128}}          
    \ifx#1Q
     \else\put(588,0)   {\line(1,1)   {100}}          
          \put(700,75)  {#1}                     \fi  
    \ifx#2Q
     \else\put(588,0)   {\line(1,-1)  {100}}          
          \put(700,-120){#2}                     \fi  
    \ifx#3Q
     \else\put(429,-159){\line(0,1)   {85}}           
          \put(225,-75) {\makebox(250,87)[r]{#3}} \fi 
    \ifx#4Q
     \else\put(429,-159){\line(0,-1)  {85}}           
          \put(400,-315){#4}                      \fi 
    \ifx#5Q
     \else\put(159,-159){\line(0,1)   {85}}           
          \put(130,-73) {#5}                     \fi  
    \ifx#6Q
     \else\put(159,-159){\line(0,-1)  {85}}           
          \put(130,-315){#6}                     \fi  
    \ifx#7Q
     \else\put(0,0)     {\line(0,1)   {85}}           
          \put(-270,84) {\makebox(300,87)[r]{#7}} \fi 
    \ifx#8Q
     \else\put(0,0)     {\line(0,-1)  {85}}           
          \put(-270,-160){\makebox(300,87)[r]{#8}}\fi 
    \put(159,159)       {\line(0,1)   {85}}           
    \put(130,250)       {\small ${\rm CH_{2}}$}       
    \put(-370,245)      {\makebox(500,87)[r]{#9}}     
  \end{picture}         }                 
 \newcommand{\furanose}[8]   {
   \begin{picture}(\pw,\pht)(-\xi,-\yi)
    \put(448,0)              {\line(-1,-2) {89}}      
    \put(359,-179)           {\line(-1,0)  {270}}     
    \put(89,-179)            {\line(-1,2)  {89}}      
    \put(0,0)                {\line(5,3)   {188}}     
    \put(192,110)            {\small O}               
    \put(260,113)            {\line(5,-3)  {188}}     
    \ifx#1Q
     \else\ifx#1N\put(448,0) {\line(0,1)   {380}}     
          \else\put(448,0)   {\line(5,3)   {105}}     
               \put(558,50)  {#1}                 \fi 
    \fi
    \ifx#2Q
     \else\put(448,0)        {\line(5,-3)  {105}}     
          \put(558,-90)      {#2}                \fi  
    \ifx#3Q
     \else\put(359,-179)     {\line(0,1)   {85}}      
          \put(155,-95)      {\makebox(250,87)[r]{#3}}\fi 
    \ifx#4Q
     \else\put(359,-179)     {\line(0,-1)  {85}}      
          \put(330,-335)     {#4}                \fi  
    \ifx#5Q
     \else\put(89,-179)      {\line(0,1)   {85}}      
          \put(60,-93)       {#5}                \fi  
    \ifx#6Q
     \else\put(89,-179)      {\line(0,-1)  {85}}      
          \put(60,-335)      {#6}                \fi  
    \ifx#7Q
     \else\put(0,0)          {\line(0,-1)  {85}}      
          \put(-270,-160)    {\makebox(300,87)[r]{#7}}\fi 
    \put(0,0)                {\line(0,1)   {85}}      
    \put(-30,90)             {\small ${\rm CH_{2}}$}
    \put(-530,80)            {\makebox(500,87)[r]{#8}}
  \end{picture}        }          
 \newcommand{\purine}[9]           {
    \begin{picture}(\pw,\pht)(-\xi,-\yi)
      \put(0,0)      {\line(0,1)   {160}}          
      \put(0,0)      {\line(5,-3)  {140}}          
      \put(342,0)    {\line(-5,-3) {140}}          
      \put(342,0)    {\line(0,1)   {200}}          
      \put(316,174)  {\line(0,-1)  {148}}          
      \put(342,200)  {\line(-5,3)  {171}}          
      \put(171,303)  {\line(-5,-3) {140}}          
      \put(342,200)  {\line(1,0)   {300}}          
      \put(684,0)    {\line(0,1)   {160}}          
      \put(684,0)    {\line(-5,-3) {140}}          
      \put(342,0)    {\line(5,-3)  {140}}          
      \put(-32,170)  {\small N}                    
      \put(135,-130) {\small N}                    
      \put(652,170)  {\small N}                    
      \put(477,-130) {\small N}                    
      \ifx#1Q
        \else\put(-128,277) {\line(5,-3) {100}}    
             \put(-430,234) {\makebox(300,87)[r]{#1}}
      \fi
      \ifx#2D\put(36,9)     {\line(5,-3) {100}}    
        \else\put(-7,11)    {\line(-5,-3) {128}}
             \put(7,-11)    {\line(-5,-3) {128}}
             \put(-430,-116){\makebox(300,87)[r]{#2}}
      \fi                                          
      \ifx#3Q                                      
        \else\put(171,-130) {\line(0,-1) {73}}
             \put(135,-283) {#3}
      \fi
      \ifx#4D\put(162,267)  {\line(-5,-3){110}} \fi 
      \ifx#5Q
        \else\multiput(158,303)(26,0){2} {\line(0,1) {100}}
             \put(135,410)  {#5}                
      \fi                                       
      \ifx#6Q                                   
        \else\put(171,303)  {\line(0,1)  {100}} 
             \put(150,410)  {#6}
      \fi
      \ifx#7Q                                   
        \else\put(812,277)  {\line(-5,-3){100}}  
             \put(817,250)  {#7}
      \fi
      \ifx#8D\put(658,26)   {\line(0,1)  {125}} 
        \else\put(691,11)   {\line(5,-3) {128}}
             \put(677,-11)  {\line(5,-3) {128}}
             \put(817,-110) {#8}                
      \fi                                       
      \ifx#9Q                                   
        \else\put(513,-130) {\line(0,-1) {73}}
             \put(492,-283) {#9}
      \fi
      \end{picture}      }            

 \newcommand{\fuseiv}[9]      {          
   \begin{picture}(\pw,\pht)(-\xi,-\yi)  
    \put(0,200)     {\line(5,3)  {171}}                
    \put(171,303)   {\line(5,-3) {171}}                
    \put(342,200)   {\line(0,-1) {200}}                
    \put(342,0)     {\line(-5,-3){171}}                
    \put(171,-103)  {\line(-5,3) {171}}                
    \ifx#1Q
     \else\put(171,303)   {\line(0,1)   {100}}         
          \put(150,410)   {#1}                     \fi 
    \ifx#2Q
     \else\put(342,200)   {\line(5,3)   {128}}         
          \put(475,250)   {#2}                     \fi 
    \ifx#3Q
     \else\put(342,0)     {\line(5,-3)  {128}}         
          \put(475,-100)  {#3}                     \fi 
    \ifx#4Q
     \else\put(171,-103)  {\line(0,-1)  {100}}         
          \put(150,-283)  {#4}                     \fi 
    \ifx#5D \put(36,191)  {\line(5,3)   {126}}     \fi 
    \ifx#6D \put(180,267) {\line(5,-3)  {126}}         
     \else\ifx#6Q
           \else\put(342,200)  {\line(5,-3)  {128}}    
                \put(475,100)  {#6}                \fi 
    \fi
    \ifx#7D \put(316,174) {\line(0,-1)  {148}}     \fi 
    \ifx#8D \put(306,9)   {\line(-5,-3) {126}}         
     \else\ifx#8Q
           \else\put(342,0)    {\line(5,3)   {128}}    
                \put(475,50)   {#8}                \fi 
    \fi
    \ifx#9D \put(162,-67) {\line(-5,3)       {126}}\fi 
  \end{picture}        }               
 \newcommand{\fuseup}[9]   {          
   \begin{picture}(\pw,\pht)(-\xi,-\yi)  
    \put(-171,303)        {\line(0,1)   {200}}         
    \put(-171,503)        {\line(5,3)   {171}}         
    \put(0,606)           {\line(5,-3)  {171}}         
    \put(171,503)         {\line(0,-1)  {200}}         
    \put(171,303)         {\line(-5,-3) {171}}         
    \ifx#1Q
     \else\put(-171,503)  {\line(-5,3)  {128}}      
          \put(-600,537)  {\makebox(300,87)[r]{#1}}\fi 
    \ifx#2Q
     \else\put(0,606)     {\line(0,1)   {100}}         
          \put(-19,713)   {#2}                     \fi 
    \ifx#3Q
     \else\put(171,503)   {\line(5,3)   {128}}         
          \put(304,553)   {#3}                     \fi 
    \ifx#4Q
     \else\put(171,303)   {\line(5,-3)  {128}}         
          \put(304,203)   {#4}                     \fi 
    \ifx#5D\put(-145,329) {\line(0,1)   {148}}     \fi 
    \ifx#6D\put(-135,494) {\line(5,3)   {126}}     \fi 
    \ifx#7D\put(9,570)    {\line(5,-3)  {126}}     \fi 
    \ifx#8D\put(145,477)  {\line(0,-1)  {148}}     \fi 
    \ifx#9D\put(135,312)  {\line(-5,-3) {126}}     \fi 
  \end{picture}      }                  
 \newcommand{\fuseiii}[6]    {          
   \begin{picture}(\pw,\pht)(-\xi,-\yi) 
    \put(0,200)     {\line(1,0)   {342}}               
    \put(342,200)   {\line(0,-1)  {200}}               
    \put(342,0)     {\line(-5,-3) {171}}               
    \put(171,-103)  {\line(-5,3)  {171}}               
    \ifx#1Q
     \else\put(342,200)   {\line(5,3)   {128}}         
          \put(475,250)   {#1}                     \fi 
    \ifx#2Q
     \else\put(342,0)     {\line(5,-3)  {128}}         
          \put(475,-100)  {#2}                     \fi 
    \ifx#3Q
     \else\put(171,-103)  {\line(0,-1)  {100}}         
          \put(150,-283)  {#3}                     \fi 
    \ifx#4Q
     \else\put(342,200)   {\line(5,-3)  {128}}         
          \put(475,100)   {#4}                     \fi 
    \ifx#5Q
     \else\put(342,0)     {\line(5,3)   {128}}         
          \put(475,50)    {#5}                     \fi 
    \ifx#6D \put(316,174) {\line(0,-1)  {148}}     \fi 
  \end{picture}      }                  
 \newcommand{\cto}[3]   {
   \len=50   \multiply \len by #3                
   \advance  \len by 120                         
     \begin{picture}(\pw,\pht)(-\xi,-\yi)
      \put(60,50)     {\vector(1,0)   {\len}}    
      \put(90,70)     {\makebox(\len,70)[l]
                      {\scriptsize ${\rm #1}$}}  
      \put(90,-40)    {\makebox(\len,70)[l]
                      {\scriptsize ${\rm #2}$}}  
    \end{picture}     }      
 \newcommand{\sbond}[1]    {
  \xbox=#1  \advance \xbox by 2
  \hspace{1.5pt} \parbox{\xbox pt} {\rule{#1 pt}{0.4pt} }  
  \xbox=50     }  
 \newcommand{\dbond}[2]    {
  \xbox=#1  \advance \xbox by 2
  \hspace{1.5pt}\parbox{\xbox pt} 
                       {\rule{#1 pt}{0.4pt}\vspace{-#2 pt}\\
                        \rule{#1 pt}{0.4pt}  }
  \xbox=50     }  
 \newcommand{\tbond}[2]   {
  \xbox=#1   \advance  \xbox by 2
  \hspace{1.5pt} \parbox{\xbox pt} 
                        {\rule{#1 pt}{0.4pt}\vspace{-#2 pt}\\
                         \rule{#1 pt}{0.4pt}\vspace{-#2 pt}\\
                         \rule{#1 pt}{0.4pt}   }
  \xbox=50     }   